\documentclass{article}
\usepackage{graphicx} 
\usepackage{pgf-umlsd} 
\usepackage{amssymb}
\usepackage{amsmath} 
\usepackage{enumitem}
\usepackage{hyperref}
\usepackage{url}
\setlist[itemize]{noitemsep} 
\setlist[enumerate]{noitemsep}

\newtheorem{definition}{Definition}[section]
\newtheorem{theorem}{Theorem}[section]
\newtheorem{lemma}{Lemma}[section]

\newenvironment{proof}{\textsf{Proof}.}{\hfill$\Box$}

\newcommand{\bits}[1]{{\{0,1\}^{#1}}}

\newcommand{\keygen}[0]{\mathsf{G}}
\newcommand{\sig}[0]{\mathsf{S}}
\newcommand{\sigver}[0]{\mathsf{V}}

\newcommand{\pubkey}[0]{\mathsf{pk}}
\newcommand{\prikey}[0]{\mathsf{sk}}
\newcommand{\mpubkey}{\widetilde{\pubkey}}

\newcommand{\hfgen}[0]{\mathsf{G}} 
\newcommand{\hfunc}[0]{\mathsf{H}} 

\newcommand{\setup}[0]{\mathsf{Set}}
\newcommand{\commit}[0]{\mathsf{Com}}
\newcommand{\open}[0]{\mathsf{Open}}
\newcommand{\commitc}[0]{{\mathsf{Com}^{c}}}

\newcommand{\param}[0]{\mathsf{par}}

\newcommand{\unisrv}[0]{\mathsf{US}}

\newcommand{\sthash}[0]{{h_\mathsf{st}}}
\newcommand{\txhash}[0]{{h_\mathsf{tx}}}

\newcommand{\auxd}[0]{\mathsf{aux}}

\newcommand{\univer}[0]{\mathcal{V}}

\newcommand{\certver}[0]{{\mathcal{V}_\mathsf{cert}}}

\newcommand{\pinc}[0]{{\pi_{\mathsf{inc}}}}

\newcommand{\prob}[0]{\mathsf{Pr}} 

\title{ The Unicity Execution Layer }
\author{
Ahto Buldas$^{1}$ \and
Dirk Draheim$^{2}$ \and
Mike Gault$^{3}$ \and
Risto Laanoja$^{3}$ \and
Vladimir Rogojin$^{3}$ \and
Ahto Truu$^{3}$
\\[1em]
$^{1}$ Tallinn University of Technology, Estonia, ahto.buldas@taltech.ee\\
$^{2}$ Tallinn University of Technology, Estonia, dirk.draheim@taltech.ee\\
$^{3}$ Unicity Labs OÜ, Estonia, ahto.truu@unicity-labs.com
}
\date{\today}

\begin{document}

\maketitle

\begin{abstract}
This paper introduces the Unicity Execution Layer, a modular component of the Unicity framework enabling secure off-chain transactions while maintaining trustless double-spending prevention. We present a formal security model where token ownership is represented by public keys and transfers require digital signatures. We prove three fundamental security properties: (1) no double-spending--each token state can be spent at most once, (2) no blocking--only the legitimate owner can prevent a token from being spent, and (3) service-side privacy--the Unicity Service cannot link transactions with the same token. The user-side privacy is addressed by introducing generalized multi-public-key signature schemes that allow one secret to generate multiple unlinkable public keys, and interactive and non-interactive concrete instantiations, enabling private transactions with stable public identity with minimal key management overhead.
\end{abstract}

\section{Introduction}

Blockchain technology has revolutionized digital asset management by enabling trustless peer-to-peer transactions without relying on centralized authorities. However, traditional blockchain architectures face fundamental scalability limitations that hinder their adoption for high-throughput applications. The core bottleneck stems from the fact that the ``security'' depends on the number of participating validators, which all have to participate in consensus on ordering, re-execute transactions, and store every produced block.

This paper introduces \emph{Unicity}, a novel blockchain infrastructure designed to enable secure off-chain transactions while maintaining the trustless guarantees of traditional blockchains. The key insight underlying Unicity is that the vast majority of blockchain operations—transaction execution, smart contract processing, and state transitions—can be moved off-chain, leaving only the essential double-spending prevention mechanism on-chain. This also simplifies on-chain operations, making efficient and self-authenticating implementations possible.

By minimizing the data that must be processed by the consensus layer, Unicity achieves linear scalability while preserving the security properties that make blockchains trustworthy. The system consists of three hierarchical layers: the Consensus Layer provides decentralized agreement and cryptoeconomical incentives, the Aggregation Layer maintains a distributed append-only dictionary of spent token states, and the Execution Layer handles peer-to-peer transaction processing and business logic.

Our approach differs fundamentally from existing scaling solutions. Rather than optimizing transaction throughput within the constraints of traditional blockchain architectures, Unicity reconceptualizes the shared server-side functionality as a minimal, trustless service which prevents double-spending. This architectural shift enables transactions to occur off-chain and, with hardware-based unicity-proving functionality, completely offline, while maintaining cryptographic guarantees against fraud.

The contributions of this paper include: (1) a formal security model for off-chain transactions with on-chain double-spending prevention, modeled as trusted service in this paper's scope, (2) cryptographic protocols ensuring transaction privacy and preventing attacks of blocking token spending, and (3) formal proofs of these security properties.

\paragraph{Paper Structure}
After introduction, Sections~\ref{sec:infrastructure}--\ref{sec:privacy} present the core Unicity infrastructure with signature-based token ownership ($\sigver(\pubkey, m, \sigma) = 1$), proving three core security properties: no double-spending, no blocking, and service-side privacy (transaction unlinkability).

Section \ref{sec:wallet-privacy} addresses user-side privacy through multi-public-key (MPK) signature schemes: a theoretical framework where one secret generates multiple unlinkable public keys, followed by concrete instantiation for ECDSA, and a protocol enabling efficient private transactions with persistent public identity.

\section{System Overview}\label{sec:overview}

\subsection{Motivation}

\begin{figure}[!htbp]
 \begin{minipage}[h]{0.55\linewidth}
        \centering
                \includegraphics[width=\columnwidth]{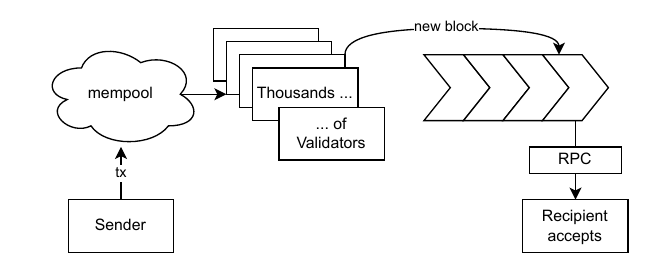}
                \caption{Data flow of a typical blockchain.} \label{fi:traditional}
    \end{minipage}
    \hfill
    \begin{minipage}[h]{0.44\linewidth}
        \centering
                \includegraphics[width=\columnwidth]{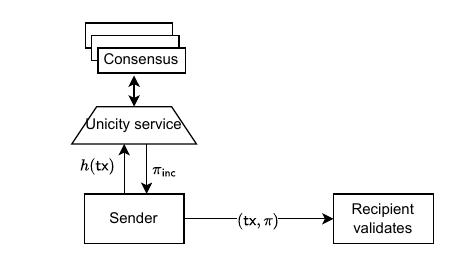}
                \caption{Data flow of Unicity transactions.}\label{fi:unicity}
    \end{minipage}
\end{figure}

Traditional blockchain architectures, illustrated in Figure~\ref{fi:traditional}, require every validator node to process all transactions sequentially. This design creates several fundamental bottlenecks: (1) \emph{computational overhead} from validating every transaction, (2) \emph{storage requirements} that grow linearly with transaction history, and (3) \emph{bandwidth limitations} from broadcasting all transaction data to every node. These constraints result in throughput limitations measured in tens of transactions per second for major blockchain networks, and transaction processing latency (time to finality) which is not suitable for interactive use cases.

Existing scaling approaches attempt to optimize within these architectural constraints. Layer-2 solutions batch transactions but still require periodic settlement on the main chain. Sharding distributes computation, but introduces complex cross-shard communication protocols. Both approaches face fundamental trade-offs between decentralization, security, and scalability.

Unicity takes a fundamentally different approach by recognizing that most blockchain operations can be moved off-chain and performed by the party who is naturally interested in the validity of the transaction, the recipient (relying party). The key insight is that central coordination is required only to prevent double-spending---the creation of multiple valid transactions spending the same digital asset. Other functions, including transaction execution, smart contract processing, state updates, and data availability, can be provided by interested parties without global agreement.

Figure~\ref{fi:unicity} illustrates the Unicity transaction flow. Rather than broadcasting full transaction data to all network participants, Unicity maintains only a cryptographic commitment to spent asset states.

Unicity as a transacting framework provides three essential guarantees: (1) \emph{unique spending}—a digital asset can be spent no more than once, (2) \emph{non-blocking}—only the legitimate owner of an asset can mark this asset as spent, and (3) \emph{privacy}—transaction details remain confidential between participants, hidden from the Unicity Service.

By decoupling transaction execution from consensus, Unicity enables new use cases previously impractical on traditional blockchains. Transactions can occur entirely off-chain, requiring no network connectivity at the time of execution. Multiple parties can transact directly using any communication channel, from internet protocols to physical media exchange. The resulting system scales linearly with the number of participants rather than facing the quadratic complexity growth of traditional blockchain networks.

\subsection{Architecture}

Unicity employs a hierarchical architecture that provides top-to-bottom decentralization and scalability, as illustrated in Figure~\ref{fig:layers}.

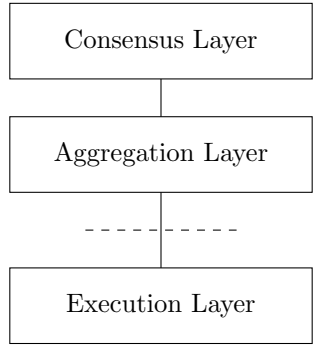
\begin{figure}[!htbp]
    \centering
        \begin{tikzpicture}
            \draw (0,0) rectangle (4,1) node[midway] {Consensus Layer};
            \draw (2,0) -- (2,-0.5);
            \draw (0,-1.5) rectangle (4,-0.5) node[midway] {Aggregation Layer};
            \draw[dashed] (1,-2) -- (3,-2);
            \draw (2,-1.5) -- (2,-2.5);
            \draw (0,-3.5) rectangle (4,-2.5) node[midway] {Execution Layer};
        \end{tikzpicture}
    \caption{Layered, hierarchical architecture of the Unicity Network.}\label{fig:layers}
\end{figure}

The three layers serve distinct functions:

\begin{itemize}
\item \textbf{Consensus Layer} provides decentralized agreement and finality through a combination of Proof-of-Work mining, providing robust decentralization, and BFT consensus with fast and deterministic finality. This layer verifies the integrity of the Aggregation Layer's state transitions and serves as the root of trust for the entire system.

\item \textbf{Aggregation Layer} implements the Unicity Service, maintaining a global append-only registry of spent token states. It provides inclusion and non-inclusion proofs, processes state certification requests, and with these services allows Execution Layer to avoid the risk of double-spending. The layer is sharded for scalability, clustered for high availability, and uses cryptographic consistency proofs to maintain trustless operation.

\item \textbf{Execution Layer} handles transaction processing, smart contract execution (implemented through orchestrated execution of programmable stateful spending conditions, called \emph{predicates}, discussed in a follow-up paper\cite{predicates}; and business logic. This layer operates off-chain and is managed by users and agents who are interested parties in transaction validation and ordering.
\end{itemize}

\subsection{Unicity Service Protocol}

The Unicity execution framework relies on the \emph{Unicity Service} that maintains a global, append-only registry of spent token states. Each digital token has an associated state hash that uniquely identifies its current ownership and transaction history. When a token owner wishes to transfer ownership, they create a signed transaction that references the current state and specifies the new owner.

\paragraph{Protocol Participants}
Within this architecture, the protocol involves three entities:

\begin{itemize}
\item \textbf{Token Owners} possess digital assets represented as tokens with unique state hashes. Owners sign transactions to transfer ownership and request certification from the Aggregation Layer.

\item \textbf{Unicity Service} (provided by the Aggregation Layer) maintains a data store, modeled in this paper as key-value store $R$ where keys are derived from public keys and state hashes, and the values are transaction hashes. The service accepts certification requests and provides inclusion proofs for registered transactions.

\item \textbf{Recipients} are the relying parties who receive token transfers and must verify the authenticity of transactions cryptographically before accepting ownership.
\end{itemize}

\paragraph{Transaction Structure}
Each transaction $T = (\sthash, D)$ consists of:
\begin{itemize}
\item $\sthash$: the current state hash of the token being transferred
\item $D = (\pubkey', x, \auxd')$: transaction data containing the recipient's public key $\pubkey'$, a random nonce $x$, and auxiliary data of the next state $\auxd'$
\end{itemize}

\noindent To prevent information leakage, the transaction data $D$ is committed using a perfectly hiding commitment scheme, producing a transaction data hash $\txhash = \commitc(H(D))$. The sender signs $H(\sthash, \txhash)$ and submits a certification request $Q = (\pubkey, \sthash, \txhash, \sigma)$ to the Unicity Service.

\paragraph{Double-Spending Prevention}
The Unicity Service processes certification requests by checking that (1) the digital signature is valid and (2) the key $H(\pubkey, \sthash)$ has not been previously registered. If both conditions hold, the service records the mapping $R[H(\pubkey, \sthash)] \gets \txhash$ and returns an inclusion proof $\pinc$. This mechanism ensures that each token state can be spent at most once.

The transaction flow is illustrated by the sequence diagram (Fig.~\ref{fi:unicity-transaction}).

\begin{figure}[!htb]
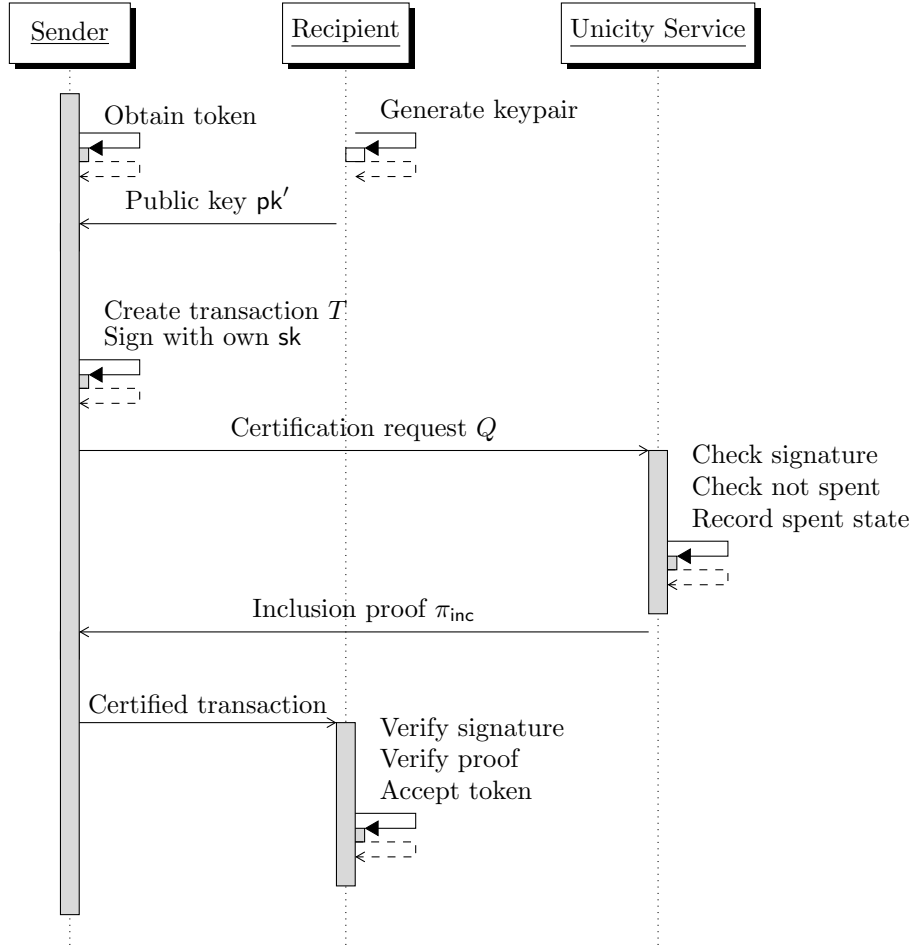

    \begin{center}
  \begin{sequencediagram}
      \newthread{S}{Sender}
      \newinst[2]{R}{Recipient}
      \newinst[2]{US}{Unicity Service}

      \begin{call}{S}{Obtain token}{S}{}
      \end{call}

      \prelevel
      \prelevel
      \begin{call}{R}{Generate keypair}{R}{}
      \end{call}

      \begin{messcall}{R}{Public key $\pubkey'$}{S}
      \end{messcall}

      \postlevel
      \begin{call}{S}{\shortstack[l]{
                Create transaction $T$\\
                Sign with own $\prikey$}}{S}{}
      \end{call}

      \begin{messcall}{S}{Certification request $Q$}{US}
          \postlevel
          \begin{call}{US}
            {\shortstack[l]{
              Check signature\\
              Check not spent\\
              Record spent state}
            }{US}{}
          \end{call}
      \end{messcall}

      \prelevel
      \begin{messcall}{US}{Inclusion proof $\pinc$}{S}
      \end{messcall}

      \begin{messcall}{S}{Certified transaction}{R}
          \postlevel
          \begin{call}{R}
            {\shortstack[l]{
              Verify signature\\
              Verify proof\\
              Accept token}
            }{R}{}
          \end{call}
      \end{messcall}
  \end{sequencediagram}
  \caption{Simplified Unicity transaction flow.}\label{fi:unicity-transaction}
    \end{center}
\end{figure}

The formal analysis that follows demonstrates that this construction provides strong security guarantees against both double-spending and blocking attacks while preserving transaction unlinkability.

\subsection{State of the Art}

\noindent\textbf{Layer 2 rollups} are secondary protocols that are intended to solve the scalability and fee issues in base (Layer 1) blockchains. 
These protocols are run by any party in parallel with the base blockchain without compromising the overall security. Layer 2 networks can process large volumes of transactions off-chain in batches and then  communicate a summary digest of the batch to the base layer, which is relatively easy to verify for the base blockchain. This saves the computing power and reduces the fees to be paid in the base blockchain. There are two types of Layer 2 rollups:
\begin{itemize}
\item \emph{Optimistic rollups} \cite{CCAM24,BBBB22,Opti26,Base26} in which Layer 1 assumes by default that all committed transactions are valid and offer an arbitration protocol for detecting and proving fraudulent transactions later. 
\item \emph{Zero-knowledge rollups} \cite{Starknet,zkSync,Linea} in which cryptographic computational integrity proof for a batch of transactions is presented to Layer 1. The proof is easy to verify in Layer 1, but its generation may be a resource consuming computational process. 
\end{itemize}
One of the drawbacks of Layer 2 networks is the settlement time (to process and confirm a transaction). Although, rollups mostly offer  soft confirmations of transactions in just a few seconds, the true Layer 1 finality can take a few minutes up to a week depending on the network. In the zk-rollups, the proof presented to Layer 1 has to involve the verification of all rules of  transaction processing and therefore, the proof generation is resource consuming. 
\medskip

\noindent\textbf{RGB smart contracts} \cite{RGB,Ihan24,LearnRGB,IntroRGB,AllRGB} are a private Layer 2/Layer 3 system for Bitcoin and the Lightning Network. Instead of storing data on the base blockchain, data is processed off-chain, so that the state history of contracts and data are kept off-chain and validated only by users interested in a particular contract (so called \emph{client-side validation}). Smart contract states are locked to specific Bitcoin UTXOs (Unspent Transaction Outputs) which have to be spent during the next transaction with the asset related to the smart contract. This technique prevents double-spending and is referred to as \emph{single-use seals}.  
The key benefits are privacy (as data is held off-chain, it is not possible to access by third parties) and Bitcoin compatibility 
(users can create their own tokens and rely on Bitcoin’s security and the speed of the Lightening Network without the need to create a new blockchain).  

A main drawback of RGB smart contracts is that the sender and the receiver must both be online and interact directly while performing a transfer — the receiver generates and shares a UTXO invoice before the asset can be sent securely. 
Hence, the transfers are not fully asynchronous which may create friction compared to traditional blockchains and hence, implementing decentralized applications  remains difficult. 
Compared to the Unicity framework, RGB lacks a consensus-anchored global non-inclusion oracle and formal exact-security proofs.
\medskip

\noindent\textbf{CoinJoin and Stealth Addresses}\cite{CoinJoin} are techniques for ensuring privacy of blockchain transactions. CoinJoin obscures the link between senders and receivers by combining inputs and outputs into a single transaction, while stealth addresses (invented during the Dark Wallet project around 2013) solve the problem of address reuse. 
Both methods aim to conceal 
the chain of ownership history of the same asset in the blockchain. 
In the Unicity framework, this goal is achieved by design as we will prove in Sec.\ref{sec:privacy}. 
\medskip

\noindent\textbf{Cryptographic Accumulators} \cite{BeMa93,CaLy02,Lipm12,BCCK25,BoBF19} aggregate a large set of data into a single, succinct digest and allow verifiers to check if an element belongs to the set using a short proof (or witness). The proof can be verified without revealing other aggregated elements.

In the Unicity framework, the Unicity Service uses a cryptographic accumulator without trusted setup currently implemented as a sparse Merkle tree. The main reason is that sparse Merkle trees is the most natural way of aggregating a function (a key-value store) rather than a set. We will show in Sec.\ref{rasaset} that while aggregating just a set prevents double-spending, but is insufficient for other security properties. 
Other types of accumulators may be considered as a future work.

\section{Preliminaries and Notation}

\subsection{Probabilities}

In the paper, we only use \emph{finite probability spaces} that are defined as pairs $(\Omega, \prob)$ so that $\Omega$ is a finite set and $\prob$ is a function from the powerset (the set of all subsets) of $\Omega$ to the interval $[0,1]$ of real numbers so that:
\begin{enumerate}
\item $\prob(\Omega)=1$
\item $\prob(A\cup B)=\prob(A) + \prob(B)$ for every $A,B\subseteq \Omega$ with $A\cap B=\emptyset$
\end{enumerate}

\noindent The set $\Omega$ is called \emph{sample set} and $\prob$ is called \emph{probability function}. The subsets of $\Omega$ are called \emph{events}. For the probability $\prob[\{\omega\}]$ of a singleton subset, we use the shorthand notation $\prob[\omega]$.
By \emph{random variable} we mean any function $X\colon \Omega \rightarrow R$ where $R$ is called the \emph{range} of the random variable. If $x\in R$ we use the notation $\prob[X=x] = \prob[X^{-1}(x)]$, where $X^{-1}(x)=\{\omega\in\Omega\colon X(\omega)=x\}$ is the $X$-preimage of $x$.

As $\Omega$ is finite, we can express the probability $\prob(A)$ of any event $A$ as the sum $\prob[A]=\sum_\omega \prob[w]\cdot [w\in A]$, where
$[w\in A]$ is the \emph{Iverson symbol}, i.e. $[w\in A]\in\{0,1\}$ and $[w\in A]=1$ iff $w\in A$. We also use Iverson symbol in a more general case for any mathematical statements $\mathcal{A}$ so that $[\mathcal{A}]=1$ iff $\mathcal{A}$ holds. For example, $\prob[X=x]=\sum_\omega \prob[w]\cdot [X(w)=x]$. Note that $[\mathcal{A}\wedge \mathcal{B}]=[\mathcal{A}]\cdot[\mathcal{B}]$ for any two statements
$\mathcal{A}$ and $\mathcal{B}$.

By \emph{probability distribution} of a random variable $X\colon \Omega \rightarrow R$ we mean the function $\mathcal{D}_X\colon R\rightarrow [0,1]$ such that $\mathcal{D}_X(x)=\prob[X=x]$ for every $x\in R$.
If $\mathcal{D}_X$ is a constant, i.e.
$\mathcal{D}_X(x)=\frac{1}{|R|}$ for every $x\in R$, then we say that the distribution is \emph{uniform}.
We use the notation $X\gets R$ to denote that $X$ is a uniformly distributed random variable with range $R$ and also to say that $X$ is \emph{uniformly sampled} from $R$.
A random variable $X\colon \Omega \rightarrow R$ is \emph{$t$-time sampleable} if there is a $t$-time probabilistic Turing machine $\mathsf{M}$
with all outputs values in $R$ and every output value $x\gets \mathsf{M}$ occurs with probability $\mathcal{D}_X(x)$, i.e. the output distribution of
$\mathsf{M}$ is $\mathcal{D}_X$.

If $X\colon \Omega\rightarrow R_X$ and $Y\colon \Omega\rightarrow R_Y$ are
random variables, $x\in R_X$ and $y\in R_Y$, then we use the notation
$\prob[X=x, Y=y] = \prob[X^{-1}(x)\cap Y^{-1}(y)]=\sum_\omega \prob[\omega]\cdot[X(\omega)=x\,\wedge\, Y(\omega)=y]$.
The probability distribution $\mathcal{D}_{X,Y}\colon R_X \times R_Y\rightarrow [0,1]$ defined by $\mathcal{D}_{X,Y}(x,y)= \prob[X=x, Y=y]$ is called the \emph{joint distribution} of $X$ and $Y$.
We say that $X$ and $Y$ are \emph{independent} if
\[
\prob[X=x,Y=y]=\prob[X=x]\cdot\prob[Y=y]
\]
for every $x\in R_X$ and $y\in R_Y$.
If $(\Omega_1,\prob_1)$ and $(\Omega_2,\prob_2)$ are probability spaces, then their \emph{direct product} is the probability space $(\Omega, \prob)$, such that $\Omega = \Omega_1 \times \Omega_2$ and $\prob[\omega_1,\omega_2]=\prob_1[\omega_1]\cdot \prob_2[\omega_2]$ for every $\omega_1\in\Omega_1$ and $\omega_2\in\Omega_2$. We will omit the indices of the probability functions when it will not cause confusion.

\subsection{Security and Security Proofs}

A \emph{cryptographic primitive} is described as a list of (parametrized) algorithms (finite sequences of atomic commands), correctness conditions (invariants), and attack scenarios. Adversaries are algorithms that participate in the security scenarios (interacting with environment) and break (are successful in the attack scenario of) the primitive with certain \emph{success} (\emph{advantage}) $\epsilon\in[0,1]$, which often is the probability of a certain logical condition about the attack scenario.
If the parameters of a cryptographic primitive are fixed, we get an \emph{instance} of the primitive.

The \emph{running time} of an adversary is the number of atomic commands the adversary executes during the attack scenario. We assume that the running time includes the code upload time, i.e. the running time is always greater than the size of the algorithm. This assumption is necessary when the primitives are fixed algorithms rather than parametrized families of algorithms like in asymptotic security models, e.g. the \emph{polynomial model}. In this paper we use the \emph{exact security model} that more precisely captures the practical use of cryptography where the primitives and adversaries are fixed algorithms.

Every instance $f$ of a primitive has \emph{security profile}
which is a function $S_f\colon [0,1]\rightarrow \mathbb{N}$ that for every $\epsilon\in [0,1]$ returns a lower bound $S_f(\epsilon)$ of the running time of an adversary that is able to break the primitive with success at least $\epsilon$.
Security profiles are non-decreasing, i.e.
$S_f(\epsilon)\le S_f(\epsilon')$ whenever $\epsilon\le \epsilon'$.
Therefore, every adversary that breaks a primitive $f$ with success $\epsilon$ has running time $t\ge S_f(\epsilon)$.

Sometimes an instance of a cryptographic primitive $g$ is constructed from instances $f_1, \ldots, f_m$ of other cryptographic primitives (using programming techniques).
A \emph{security reduction} (or \emph{security proof}) is a mathematical proof that the constructed primitive $g$ hash a security profile $S_g$
based on the security profiles $S_{f_1}, \ldots, S_{f_m}$ of $f_1, \ldots, f_m$, respectively.

Usually, in such a proof, it is assumed that there is an adversary $A$ with running time $t$ that breaks $g$ with success $\epsilon$ and then the adversaries $A_1,\ldots, A_m$ are constructed based on $A$ that break $f_1,\ldots,f_m$ with (some unknown) successes $\epsilon_1,\ldots,\epsilon_m$, respectively, so that inequality $\epsilon\le \epsilon_1 + \ldots + \epsilon_m$ holds.

Mostly, $A_1, \ldots, A_m$ use $A$ as black-box, i.e. either call or simulate $A$ and add some computational instructions. In this paper, we only have reductions where $A$ is called only once by every $A_i$ i.e. the running times of $A_1,\ldots,A_m$ are upper-bounded by $\tau_1(t), \ldots, \tau_m(t)$, respectively, where $\tau_i$ is the computational time overhead function for constructing $A_i$ from $A$.
Therefore, we have inequalities:
\[
\tau_1(t) \ge S_{f_1}(\epsilon_1), \quad \tau_2(t) \ge S_{f_2}(\epsilon_2), \quad \ldots\quad
\tau_m(t)\ge S_{f_m}(\epsilon_m)
\]
that imply $t\ge \min_{\epsilon_1 + \ldots + \epsilon_m = \epsilon}\max\{\tau^{-1}_1(S_{f_1}(\epsilon_1)), \ldots, \tau^{-1}_m(S_{f_m}(\epsilon_m))\}$, i.e. such reductions will prove the following security profile $S_g$ of $g$:
\begin{equation}\label{eq:lower}
S_g(\epsilon) = \min_{\epsilon_1 + \ldots + \epsilon_m = \epsilon}\max\{\tau^{-1}_1(S_{f_1}(\epsilon_1)), \ldots, \tau^{-1}_m(S_{f_m}(\epsilon_m))\}
\end{equation}
The minimum is necessary because we have to consider the worst distribution of $\epsilon_1,\ldots, \epsilon_m$ because the only fact we know about $\epsilon_i$ is that they are non-negative and their sum is $\epsilon$.
Equation (\ref{eq:lower}) implies a simpler but weaker profile $S'_g$:
\[
S'_g(\epsilon) = S_\mathsf{min}(\epsilon /m) =
 \min\{\tau^{-1}_1(S_{f_1}(\epsilon/m)), \ldots, \tau^{-1}_m(S_{f_m}(\epsilon/m))\}\enspace.
\]
Moreover, if $\tau(t)=\max\{\tau_1(t), \ldots, \tau_m(t)\}$ and $S_\mathsf{min}(\epsilon)= \min\{S_{f_1}(\epsilon), \ldots, S_{f_m}(\epsilon)\}$ then we have an even simpler security profile $S''_g$ for $g$ defined by:
\begin{equation}\label{eq:losebound}
S''_g(\epsilon) = \tau^{-1}(S_\mathsf{min}(\epsilon/m))\enspace.
\end{equation}
In the security reductions of this paper, the time overhead function $\tau$ is linear, i.e. $\tau(t) = \alpha t + \beta$, where $\alpha$ and $\beta$ are reduction-specific constants.

\subsection{Signature Schemes}

A \emph{signature scheme} is a triple $(\keygen, \sig, \sigver)$ of algorithms such that:
\begin{itemize}
\item $(\pubkey, \prikey) \gets \keygen$ generates a public key $\pubkey$ and a private key $\prikey$
\item $\sigma \gets \sig(\prikey, m)$ generates a signature on a message $m$
\item $b \gets \sigver (\pubkey, m, \sigma)$ verifies a signature on a message (accepts if $b = 1$)
\end{itemize}
so that for every message $m$ the following \emph{verification identity} holds:
\[
\mathsf{Pr}[(\pubkey, \prikey) \gets \keygen \colon\; \sigver(\pubkey, m, \sig(\prikey, m)) = 1] = 1 \enspace.
\]

\begin{definition}[EF-CMA security]
A signature scheme $(\keygen, \sig, \sigver)$ is $S$-secure against existential forgeries under adaptive chosen message attacks ($S$-secure EF-CMA) if it has $S$ as a security profile in the following attack scenario:
\begin{enumerate}
\item $(\pubkey, \prikey) \gets \keygen$;
\item $(m,\sigma)\gets A^{\mathsf{S}(\prikey;)}(\pubkey)$;
\item The attack is successful iff $\mathsf{V}(\pubkey, m, \sigma)=1$
and $A$ never queries $\mathsf{S}(\prikey;m)$.
The success $\epsilon$ of $A$ is the probability that the attack is successful.
\end{enumerate}
\end{definition}

\subsection{One-Way Functions}

Let $f\colon X \rightarrow Y$ be any function from the range $X$ to a domain $Y$.

\begin{definition}[one-wayness]
A function $f$ is $S$-secure one-way if it has $S$ as a security profile in the following attack scenario:
\begin{enumerate}
\item $x\gets X$, i.e. $x$ is chosen uniformly at random from the domain $X$;
\item $x'\gets A(f(x))$;
\item The attack is successful if $f(x')=f(x)$,
and the success $\epsilon$ of $A$ is the probability that the attack is successful.
\end{enumerate}
\end{definition}

\subsection{Hash Functions}

\noindent A \emph{hash function family} is a pair $(\hfgen,\hfunc)$ where:
\begin{itemize}
\item $\hfgen$ is a probabilistic algorithm that chooses a parameter $\param$
\item $\hfunc$ is a deterministic algorithm such that for every value of $\param$, the function $H=\hfunc(\param; \cdot)$ is of type $\bits{\ast} \to \bits{k}$.
\end{itemize}

\begin{definition}[collision-resistance]
A hash function family $(\hfgen,\hfunc)$ is $S$-secure collision-resistant if it has $S$ as a security profile in the following attack scenario:
\begin{enumerate}
\item $\param \gets \hfgen$;
\item $(m,m')\gets A(\param)$;
\item The attack is successful iff $m\neq m'$ and $\hfunc(\param;m)=\hfunc(\param;m')$,
and the success $\epsilon$ of $A$ is the probability that the attack is successful
\end{enumerate}
\end{definition}

\noindent In the following, we assume that the sampling $\param \gets \hfgen$ has been done before any attack scenario, and we often say that the function $H=\hfunc(\param; \cdot)$ itself is collision-resistant regardless of the fact that no fixed function can formally be collision-resistant.

\begin{definition}[$(k,\ell)$-one-wayness]
A function $H\colon \{0,1\}^\ast\rightarrow\{0,1\}^k$
is $S$-secure $(k,\ell)-$\emph{one-way} if it has $S$ as a security profile in the following attack scenario:
\begin{enumerate}
\item $(h,a)\gets A_1$;
\item $x \gets \{0,1\}^\ell$;
\item $x'\gets A_2(a;H(h,x))$;
\item The attack is successful iff $h\in\{0,1\}^k$, $x'\in\{0,1\}^\ell$ and $H(h,x)=H(h,x')$.
The success $\epsilon$ of $A$ is the probability that the attack is successful.
\end{enumerate}
\end{definition}

\noindent Equivalently, $H$ is $S$-secure $(k,\ell)$-one-way iff the function $f_h$ defined by $f_h(x)=H(h,x)$ is $S$-secure one-way for every $h\in\{0,1\}^k$.

\subsection{Commitment Schemes}

A \emph{commitment scheme} is a triple $(\setup, \commit, \open)$ of probabilistic algorithms such that:
\begin{itemize}
\item $\param\gets\setup$ is the setup algorithm that fixes the parameters of the scheme
\item $(c,d)\gets \commit(\param; m)$ computes commitment $c$ and
decommitment string $d$ of a message $m$
\item $m\gets \open(\param; c, d)$ opens the commitment
\end{itemize}
so that for every $m$, the following correctness identity holds:
\[
m=\open(\param; \commit(\param; m))\enspace.
\]
We denote by $\commitc(\param,m)$ the function that computes $(c,d)\gets \commit(\param; m)$ and returns $c$.
We will often omit the parameter $\param$ and use the shorthand notations $\commit(m)$, $\commitc(m)$ and $\open(c,d)$ instead of $\commit(\param; m)$,
$\commitc(\param; m)$ and $\open(\param; c, d)$, respectively.
\medskip

\begin{definition}[trivial commitment scheme]
In the trivial commitment scheme $(\setup,\commit,\open)$ the functions are defined as follows:
\begin{itemize}
\item $\setup$ always returns $\bot$.
\item $\commit(m)=(m,\bot)$ is the identity function.
\item $\open (c,d) = c$ just returns the first argument.
\end{itemize}
\end{definition}

\noindent In terms of security, the commitment schemes are required to be \emph{binding} and \emph{hiding}. The Binding property means that once the commitment $c$ is fixed, it is not possible (or very hard) to open it in two different ways. The Hiding property means that the commitment $c$ must not contain efficiently extractable information about the committed message.

\subsubsection{Binding}

\begin{definition}[binding]
A commitment scheme $(\setup, \commit, \open)$ is $S$-secure computationally binding if it has $S$ as a security profile in the following attack scenario:
\begin{enumerate}
\item $\param\gets\setup$;
\item $c,d,d'\gets A(\param)$;
\item The attack is successful if $\open(\param; c, d)\neq\open(\param; c, d')$,
and the success $\epsilon$ of $A$ is the probability that the attack is successful.
\end{enumerate}
\end{definition}

\noindent This property is called \emph{computational binding} because it protects against adversaries with limited computational power. There exist commitment schemes that are \emph{perfectly binding}, which means that opening a commitment in two different ways is impossible by definition. For example, the trivial commitment is perfectly binding, however it is ``perfectly non-hiding'' because the commitment of $m$ is $m$ itself.

\subsubsection{Hiding}

\begin{definition}[perfect hiding]
A commitment scheme $(\setup, \commit, \open)$ is said to be perfectly hiding if for every $\param$ and for every two messages $m,m'$ the commitments $c\gets \commitc(\param; m)$ and $c'\gets \commitc(\param; m')$ have equal probability distributions as random variables (assuming that the two calls of $\commit$ use independent internal random strings).
\end{definition}

\begin{lemma}[output independence]\label{le:outputindependence}
If $(\setup, \commit, \open)$ is a perfectly hiding commitment scheme, $m$ is chosen according to any probability distribution and $g$ is any deterministic function, then $m$ and $\commitc(g(m))$ are independent random variables.
\end{lemma}
\begin{proof}
Let $\omega \gets \Omega$ be the internal randomness sampling of $\commitc$. We denote by $\commitc_\omega$ the deterministic version of $\commitc$ where the internal random string $\omega$ is fixed. Let $\mathcal{M}$ denote the sampling space of $m$. We assume that the sampling $\mu\gets \mathcal{M}$ happens independently of $\omega \gets \Omega$ and hence, the total sampling space is a direct product space with sampling space $\mathcal{M}\times \Omega$ and $\prob[\mu,\omega]=\prob[\mu]\cdot \prob[\omega]$ for any values $\mu, \omega$.

Let $m=M(\mu,\omega)$ with $\mu\gets \mathcal{M}$ and $\omega\gets \Omega$, i.e. $M$ is the random variable corresponding to $m$. Let $C$ be the corresponding random variable of $c$, i.e. $c=C(\mu,\omega)=\commitc_\omega(g(M(\mu,\omega)))$. Note that $M(\mu,\omega)$ does not depend on $\omega$.

We have to show that $\prob[M=m,C=c] = \prob[M=m]\cdot \prob[C=c]$ in this probability space for any possible values $m,c$ of the message and the commitment, respectively.
\begin{eqnarray*}
\prob[M=m,C=c] &=& \sum_{\mu,\omega}\prob[\mu,\omega]\cdot [M(\mu,\omega)=m]\cdot [C(\mu,\omega)=c]\\
&=& \sum_{\mu,\omega}\prob[\mu]\cdot\prob[\omega]\cdot [M(\mu,\omega)=m]\cdot [C(\mu,\omega)=c]\\
&=& \sum_{\mu}\prob[\mu]\cdot[M(\mu,\omega)=m]\left(\sum_{\omega}\cdot\prob[\omega]\cdot [C(\mu,\omega)=c]\right)\\
&=& \sum_{\mu}\prob[\mu]\cdot[M(\mu,\omega)=m]\cdot
\prob[\commitc(g(M(\mu,\omega)))=c]\\
&=& \prob[\commit_\omega(g(M(\mu,\omega)))=c]\cdot \sum_{\mu}\prob[\mu]\cdot[g(M(\mu,\omega))=m]\\
&=& \prob[\commitc_\omega(g(M(\mu,\omega)))=c]\cdot \prob[M=m]
\end{eqnarray*}
because $p = \prob[\commitc_\omega(g(M(\mu,\omega)))=c]$ does depend neither on $\mu$ due to the perfect hiding property, nor on $\omega$. Moreover:
\begin{eqnarray*}
p &=& \prob[\commitc(g(M(\mu,\omega)))=c] \cdot \overbrace{\sum_{\mu'} \prob[\mu']}^{=1}\\
&=& \sum_{\mu'} \prob[\mu']\cdot \prob[\commitc(g(M(\mu,\omega)))=c]\\
&=& \sum_{\mu'} \prob[\mu']\cdot \prob[\commitc(g(M(\mu',\omega)))=c]\\
&=& \sum_{\mu'} \prob[\mu']\cdot \sum_{\omega}\prob[\omega]\cdot [\commitc_\omega(g(M(\mu',\omega)))=c]\\
&=& \sum_{\mu',\omega} \prob[\mu',\omega]\cdot [\commitc_\omega(g(M(\mu',\omega)))=c] = \prob[C=c]
\end{eqnarray*}
that proves the claim.
\end{proof}

\subsection{Perfectly Hiding Commitments and One-Wayness}

Let $f$ be a one way function and $(\setup, \commit, \open)$ be a perfectly hiding commitment scheme. We will show that $f$ remains hard to invert even if, in addition to the image $f(x)$, the adversary also knows the commitment $\commitc(x)$. The following lemma shows that knowing $\commitc(x)$ does not help the adversary (much) in inverting a one-way function.

\begin{lemma}
If $f$ is $S_f$-secure one-way and $(\setup, \commit, \open)$ is a perfectly hiding commitment scheme, then $f$ is $S'_f$-secure in the following attack scenario:
\begin{enumerate}
\item $x\gets X$;
\item $x'\gets A(f(x),\commitc(x))$;
\item The attack is successful if $f(x')=f(x)$;
\end{enumerate}
where $S'_f(\epsilon) = S_f(\epsilon) - t_\mathsf{sm} - t_\mathsf{com}$, where $t_\mathsf{sm}$ and $t_\mathsf{com}$ are the running times of the samplings $\cdot\gets X$ and $\cdot\gets \commitc(\cdot)$, respectively.\footnote{The overhead function is $\tau(t)=t + t_\mathsf{sm} + t_\mathsf{com}$ and its inverse $\tau^{-1}(t) = t - t_\mathsf{sm} - t_\mathsf{com}$.}
\end{lemma}
\begin{proof}
Let $f$ be an $S$-secure one-way function and $A$ be a $t$-time adversary that with probability $\epsilon$ succeeds in the attack scenario.
Consider the following modified scenario with the same adversary:
\begin{enumerate}
\item $x\gets X$;
\item $x''\gets X$;
\item $x'\gets A(f(x),\commit(x''))$;
\item the attack is successful if $f(x')=f(x)$.
\end{enumerate}

\noindent From Lemma~\ref{le:outputindependence} it follows that in both input distributions $(f(x),\commit(x))$ and $(f(x),\commit(x''))$ the commitments are independent of $x$ and are equally distributed, and hence
the joint distributions of $(f(x),\commit(x))$ and $(f(x),\commit(x''))$ are equal. It follows that the success probability of $A$ in the second scenario is also equal to $\epsilon$. Let $A'(y)$ be the adversary that, given $y=f(x)$ as input proceeds as follows:
\begin{enumerate}
\item $x''\gets X$;
\item $c\gets \commit(x'')$;
\item return $A(y,c)$.
\end{enumerate}

\noindent The adversary $A'$ inverts $f$ with probability $\epsilon$ and has a running time $t+t_\mathsf{sm} + t_\mathsf{com}$. Therefore, $t+t_\mathsf{sm} + t_\mathsf{com}\ge S_f(\epsilon)$ and hence $t\ge S_f(\epsilon)- t_\mathsf{sm} - t_\mathsf{com}$.
\end{proof}
\medskip


\subsection{Pseudo-random Function Families}

\begin{definition}[PRF] An $S$-secure \emph{pseudo-random function family (PRF)} is a function $F\colon K\times X \rightarrow Y$ that has $S$ as a security profile in the following attack scenario with a distinguisher $D$:
\begin{enumerate}
\item $k\gets K$
\item $\Phi\gets Y^X$, i.e. $\Phi$ is a randomly chosen function of type $X\rightarrow Y$
\item $b_1\gets D^{F(k;\cdot)}$
\item $b_0\gets D^{\Phi(\cdot)}$
\item The success of $D$ is
$\epsilon=|\mathsf{Pr}[b_1=1]-\mathsf{Pr}[b_0=1]|$
\end{enumerate}
\end{definition}
The oracle $\Phi$ can be simulated by using the so-called \emph{lazy sampling} technique. The oracle stores a partial function (dictionary) $\phi$
that is initially nowhere defined (i.e. $\phi[x]=\bot$ for every $x\in X$) and every oracle call $\Phi(x)$ is handled as follows:
\begin{enumerate}
\item If $\phi[x]\neq \bot$ then return $\phi[x]$.
\item If $\phi[x]=\bot$ then:
\begin{enumerate}
\item Pick a random $y\gets Y$
\item Define $\phi[x]\gets y$
\item Return $y$
\end{enumerate}
\end{enumerate}

\section{Unicity Infrastructure}\label{sec:infrastructure}

Unicity infrastructure is about maintaining identifiable digital assets called \emph{tokens}. For example, tokens can represent units of digital currency.

\emph{Parties} can \emph{create (issue)} tokens, \emph{own} tokens, and \emph{transfer} tokens to each other, i.e. the ownership of tokens may change. In order to transfer a token, its owner makes a signed \emph{transaction} that redefines the ownership.

We assume that transferred tokens can be sent using any channels and their storage does not require dedicated hardware devices. At the same time, the infrastructure has to guarantee some properties of the tokens such as \emph{unique ownership}, i.e. the owner of a token should not be able to transfer the token to two different parties (i.e. \emph{double-spend} the token), and once a token has been transferred, neither the previous owner nor any third parties should be able to do anything with the token---transfer it or make it unusable for the next owner (i.e. \emph{block} the token).

As nothing prevents copying of digital information, some additional components are needed in the infrastructure to guarantee the desired properties of tokens. For this, the Unicity infrastructure includes the \emph{Unicity Service}---an online functionality that all parties can communicate with.

In this section, we assume that $(\keygen, \sig, \sigver)$ is a signature scheme and $H$ is a hash function.

\subsection{Unicity Service}\label{sec:unicity-service}

We first model the Unicity Service as an ideal functionality, and later discuss how to implement such a service in a secure and efficient way.

The \emph{Unicity Service} $\unisrv$ is modeled as a state machine with state $R$, which is a key-value store (dictionary), where both keys and values are of type $\bits{|k|}$. Initially, $R = \emptyset$. We will write $R[k] = \bot$ if there are no pairs $(k, v)$ stored in $R$.

Every input request $Q$ is a tuple $(\pubkey, \sthash, \txhash, \sigma)$, where:
\begin{itemize}
\item $\pubkey$ of type $\bits{p}$ is a public key: the public key of the current owner of a token, i.e. the owner before the transaction with the hash $\txhash$ is executed;
\item $\sthash$ of type $\bits{k}$ is a ``state hash'', a value linking subsequent token states;
\item $\txhash$ of type $\bits{k}$ is a transaction data hash (defined later);
\item $\sigma$ of type $\bits{s}$ is a digital signature of the transaction.
\end{itemize}

\noindent The request $Q=(\pubkey,\sthash,\txhash,\sigma)$ is processed by $\unisrv$ with the state $R$ as follows:
\begin{enumerate}
\item If $R[H(\pubkey, \sthash)] = \bot$ and $\sigver(\pubkey, H(\sthash, \txhash), \sigma) = 1$ then
\[
R \gets R \cup \{(H(\pubkey, \sthash), \txhash)\} \enspace,
\]
i.e. the new value of $R$ is defined by setting $R[H(\pubkey, \sthash)] \gets \txhash$ and leaving the rest of the contents of $R$ unchanged.
\item A proof $\pinc$ of the statement $R[H(\pubkey, \sthash)] = v$ (inclusion proof) is returned.
\end{enumerate}

\noindent It is easy to see that if $R_0 = \emptyset$ is the initial state, $Q_1, Q_2, \ldots, Q_n$ is any sequence of queries, and $R_i$ is the state after the request $Q_i$ then:
\begin{itemize}
\item $R_0 \subseteq R_1 \subseteq \ldots \subseteq R_n$, i.e. the elements are never removed.
\item The state $R_n$ is a \emph{partial function}, i.e.
$\{(k, v), (k, v')\} \subseteq R$ implies $v = v'$.
\end{itemize}

We say that a key $k$ is \emph{blocked} if $R[k] \neq \bot$.

\subsection{Verification Function}

We assume that the inclusion proofs $\pinc$ can be verified by any party using a verification function $\univer$ so that:
\begin{itemize}
\item If $\univer(k, v; \pinc) = 1$ then $R[k] = v$ in the current state $R$ of $\unisrv$. Hence, as $R$ is a partial function, for every $k, v, v', \pinc, \pinc'$, the following implication holds:
\begin{equation}\label{eq:eqtx}
\univer(k, v; \pinc) = \univer(k, v'; \pinc') = 1 \quad \Rightarrow \quad v =v' \enspace.
\end{equation}
\item If $R[H(\pubkey, \sthash)] = \txhash$ after a request $\pinc \gets \unisrv(\pubkey, \sthash, \txhash, \sigma)$ to the Unicity Service, then $\univer(H(\pubkey, \sthash), \txhash, \pinc) = 1$.
\end{itemize}

\subsection{Transactions with a Token}

Every token has a \emph{state hash} $\sthash$ and an owner $A$ represented by a public key $\pubkey$. The state hash $\sthash$ is initialized by the \emph{mint transaction} of the token (Sec.~\ref{sec:mint-transaction}).
We will call the pair $(\pubkey, \sthash)$ the \emph{state} of the token.
Every (unsigned) transaction with the token is a pair
$T = (\sthash, D)$, where:
\begin{enumerate}
\item $\sthash$ is the state hash linking the token ledger,
\item $D$ (transaction data) contains the following fields:
\begin{itemize}
\item $\pubkey'$: the public key of the next owner,
\item $x$: a uniformly chosen random string $x\gets\{0,1\}^\ell$,
\item $\auxd'$: other data for the next state.
\end{itemize}
\item The next state hash $\sthash'$ is computed by $\sthash'\gets H(\sthash,x)$. The pair $(\pubkey', \sthash')$ defines the next state of the token after executing the transaction $T$.
\end{enumerate}
The main idea of the state is that the transaction $T$ with a token is possible only if its current state $(\pubkey,\sthash)$ is \emph{not spent}, i.e. $R[H(\pubkey,\sthash)]=\bot$. When executed, $T$ \emph{spends} the state $(\pubkey,\sthash)$ by sending a request $Q=(\pubkey,\sthash,\txhash,\sigma)$ to $\unisrv$, i.e. $R[H(\pubkey,\sthash)]=\txhash\neq\bot$ after execution and hence, no other transactions in the same state are possible.

The next state $(\pubkey',\sthash')$ should be a non-spent state, i.e. $R[H(\pubkey',\sthash')]=\bot$ for the next transactions with the same token being possible. This is guaranteed by the one-wayness and collision-resistance of the hash function $H$. As $x$ is chosen randomly and is not visible by $\unisrv$ and moreover, it is protected by a perfectly hiding commitment scheme, it is not possible for $\unisrv$ to associate the current state and the next state of the token.
\medskip

\noindent\textbf{Certifying a transaction} $T = (\sthash, D)$ involves the following steps:
\begin{enumerate}
\item $(\txhash,d) \gets \commit(H(D))$ is computed using a perfectly hiding commitment scheme $(\setup,\commit,\open)$.
The commitment $\txhash$ is called the \emph{transaction data hash}.
\item The hash value $h_T = H(\sthash, \txhash)$ is computed.
\item A digital signature $\sigma \gets \sig(\prikey, h_T)$ is created with the private counterpart $\prikey$ of $\pubkey$, i.e. $\mathsf{V}(\pubkey, h_T, \sigma) = 1$.
\item The request $Q = (\pubkey,\sthash,\txhash,\sigma)$ is created.
\item $\unisrv$ is called to obtain $\pi\gets \unisrv(Q)$.
\item The certified transaction $(T,\sigma,\txhash,d,\pi)$ is formed.
\end{enumerate}

\noindent\textbf{Verifying a certified transaction} A certified transaction $(T,\sigma,\txhash,d,\pi)$ is verified in the state $(\pubkey,h)$ by the following algorithm:\medskip

\noindent $\certver(T,\sigma,\txhash,d,\pi;\pubkey,h)$:
If at least one of the following checks fail, return 0, otherwise return 1:
\begin{enumerate}
\item $T.\sthash=h$;
\item $\open(\txhash,d)=H(T.D)$;
\item $\sigver(\pubkey, H(\sthash,\txhash),\sigma)=1$;
\item $\univer(H(\pubkey,T.\sthash),\txhash,\pi)=1$.
\end{enumerate}

\begin{definition}[certification in a state]\label{de:certstate}
A tuple $(T,\sigma,\txhash,d,\pi)$ is said to be \emph{certified in state} $(\pubkey,h)$ iff $\certver(T,\sigma,\txhash,d,\pi;\pubkey,h)=1$.
\end{definition}

\subsection{Mint Transaction}\label{sec:mint-transaction}

Mint transaction is the first transaction with every token. Mint transaction assigns a unique \emph{Token Identifier} $\mathsf{id}$ and some more application-specific data fields, like a \emph{Mint Justification}, packed into the auxiliary data $\auxd$.
Minting uses the following public system-specific constants:
\begin{itemize}
\item $\mathsf{MINT\_SUFFIX}$ -- a fixed domain separator
\item $\pubkey_\mathsf{mint}$ -- minting public key
\item $\prikey_\mathsf{mint}$ -- minting private key
\end{itemize}
Note that $\prikey_\mathsf{mint}$ is also public and is needed only for having a unified interface with $\unisrv$.

A certified mint transaction is $(T_0, \sigma_0, \pi_0)$, where $T_0 = (\sthash, D_\mathsf{mint})$, $\sthash = H(\mathsf{id}, \mathsf{MINT\_SUFFIX})$ where $D_\mathsf{mint}$ contains the following fields:
\begin{itemize}
    \item $\pubkey'$: the public key of the first owner;
    \item $\mathsf{id}$: the token identifier;
    \item $\auxd'$: other data of the first state.
\end{itemize}
\medskip

\noindent\textbf{Certifying a mint transaction} $T = (\sthash, D_\mathsf{mint})$ involves the following steps:
\begin{enumerate}
	\item $\txhash \gets H(D_\mathsf{mint})$ (perfectly hiding commitment is unnecessary for mint)
	\item $h_T \gets H(\sthash, \txhash)$
	\item $\sigma \gets \sig(\prikey_\mathsf{mint}, h_T)$, i.e. create a digital signature the private key $\prikey_\mathsf{mint}$.
	\item $Q \gets (\pubkey_\mathsf{mint},\sthash,\txhash,\sigma)$, i.e. a request is created.
	\item $\pi\gets \unisrv(Q)$, i.e. $\unisrv$ is called to obtain an inclusion proof.
	\item Output $(T,\sigma,\pi)$ as a certified mint transaction
\end{enumerate}

\noindent\textbf{Verifying a certified mint transaction} $(T,\sigma,\pi)$ involves the following checks:
\begin{enumerate}
	\item $T.\sthash= H(T.D_\mathsf{mint}.\mathsf{id}, \mathsf{MINT\_SUFFIX})$
	\item $\sigver(\pubkey_\mathsf{mint}, H(\sthash,\txhash),\sigma)=1$, where $\txhash = H(T.D_\mathsf{mint})$
	\item $\univer(H(\pubkey_\mathsf{mint},\sthash),\txhash,\pi)=1$.
\end{enumerate}

\noindent Application-specific checks (e.g., validation of the mint authorization based on the enclosed mint justification) follow.

\subsection{Token Ledger}

A \emph{token ledger} is a sequence
\[
(T_0, \sigma_0, \pi_0; \pubkey^0, h^0_\mathsf{st}), (T_1, \sigma_1,\txhash^1,d_1,\pi_1; \pubkey^1, h^1_\mathsf{st}), \ldots, (T_n, \sigma_n, \txhash^n, d_n,\pi_n; \pubkey^n,h^n_\mathsf{st})
\]
where:
\begin{enumerate}
\item $(T_0, \sigma_0, \pi_0)$ is a certified mint transaction
\item $\pubkey^0=T_0.D_\mathsf{mint}.\pubkey'$
\item $\sthash^{0}=H(T_0.D_\mathsf{mint}.\mathsf{id},\mathsf{MINT\_SUFFIX})$
\item For every index $i=1,\ldots, n$:
\begin{enumerate}
\item[3.3.] $(T_{i}, \sigma_{i}, \txhash^i,d_i, \pi_{i})$ is a certified transaction in the state $(\pubkey^{i-1}, \sthash^{i-1})$
\item[3.1.] $\pubkey^i = T_{i}.D.\pubkey'$
\item[3.2.] $h^{i}_\mathsf{st}=H(h^{i-1}_\mathsf{st},T_{i}.D.x)$
\end{enumerate}
\end{enumerate}

\section{Security}\label{sec:security}
Consider a token with the state $S = (\pubkey, \sthash)$. The transfer protocol ensures the following properties:
\begin{itemize}
\item \emph{No blocking}:
Only the owner of the private key of $\pubkey$ can block the state $S = (\pubkey, \sthash)$ if it was not blocked before.
\item \emph{No double-spending}:
Only one certified transaction can be created in the state $S$.
\end{itemize}

\noindent In this section, we present security proofs for both the no blocking and the no double-spending properties. Security against blocking does not depend on the choice of the commitment scheme and security against double spending assumes computational binding of the commitment scheme.

Therefore, both proofs are also valid if the commitment scheme is trivial (i.e. $\txhash = H(D)=H(\pubkey',x,\auxd')$) because the trivial commitment scheme is perfectly (and hence also computationally) binding. Later when we prove the privacy properties, we have to assume that the commitment scheme is perfectly hiding.

\subsection{Security against Blocking}\label{sec:blocking}

A blocking adversary $A$ uses two oracles:
\begin{enumerate}
\item $\mathsf{US}$: the Unicity Service,
\item $\mathsf{TS}(\prikey,\cdot)$: the transaction signer that, given as input a transaction $(h,D)$ returns $(\sigma, \txhash, d)$, where
$(\txhash, d)\gets \commit(H(D))$ and
$\sigma \gets \mathsf{S}(\prikey,H(h,\txhash))$.
\end{enumerate}

\noindent\textbf{Blocking scenario} involves the following steps:
\begin{enumerate}
\item $(\pubkey,\prikey)\gets \mathsf{G}$, i.e. a keypair is generated;
\item $h_\mathsf{st}\gets A^{\mathsf{US},\mathsf{TS}(\prikey,\cdot)}(\pubkey)$, i.e. $A$ outputs a hash value;
\item $A$ is successful if
$R[H(\pubkey,h_\mathsf{st})]\neq\bot$ after the scenario and no queries of the form $(\sigma,\txhash,d)\gets\mathsf{TS}(\prikey;h_\mathsf{st},D))$ were made.
The success $\epsilon$ of $A$ is the probability that the attack is successful
\end{enumerate}

\noindent Note that if such a query was made, then the request $Q=(\pubkey,\sthash, \txhash,\sigma)$ to $\mathsf{US}$ will trivially ensure $R[H(\pubkey,h_\mathsf{st})]\neq\bot$, and hence this is excluded by the security condition. \medskip

\begin{definition}[blocking security]
The Unicity Service is said to be $S$-secure against blocking if it has $S$ as a security profile in the blocking scenario.
\end{definition}

\noindent\textbf{Analysis}: The adversary $A$ can be successful in the following cases:
\begin{itemize}
\item[a)] A request $Q=(\pubkey', h'_\mathsf{st}, h_\mathsf{tx}, \sigma)$ with $(\pubkey',h'_\mathsf{st})\neq (\pubkey,h_\mathsf{st})$ to $\mathsf{US}$ enforces $R[H(\pubkey, h_\mathsf{st})]\neq\bot$, which means that  $H(\pubkey,h_\mathsf{st})=H(\pubkey', h'_\mathsf{st})$ and hence, a collision for $H$ was found.
\item[b)] A request $Q=(\pubkey, h_\mathsf{st}, h_\mathsf{tx}, \sigma)$ to $\mathsf{US}$ enforces $R[H(\pubkey, h_\mathsf{st})]\neq\bot$, which implies $\mathsf{V}(\pubkey, H(h_\mathsf{st},h_\mathsf{tx}), \sigma)=1$ from the description of $\mathsf{US}$. Then we have two possibilities:
\begin{itemize}
\item[b1)] A request $(\sigma',\txhash',d)\gets\mathsf{TS}(\prikey;\sthash',D)$ with $\sthash'\neq h_\mathsf{st}$ was made such that
$H(\sthash',\txhash')=H(\sthash,\txhash)$, which means that a collision for $H$ was found.
\item[b2)] If no requests $(\sigma',\txhash',d)\gets\mathsf{TS}(\prikey;\sthash',D)$ were made with $H(\sthash',\txhash')=H(\sthash,\txhash)$ then this means that $A$ was able to create the signature $\sigma$ without ``help'' from the $\mathsf{TS}(\prikey;\cdot)$ oracle, and hence $A$ was able to create an existential forgery against the signature scheme.
\end{itemize}
\end{itemize}

\begin{theorem}
If the signature scheme is $S$-secure EF-CMA and the hash function is $S$-secure collision-resistant, then the Unicity service is
$S_\mathsf{block}$-secure against blocking, where $S_\mathsf{block}(\epsilon) =
\frac{S(\epsilon/2)}{\max\{t_\mathsf{ver}, t_\mathsf{sig}\}+1} - \frac{t_\mathsf{gen}}{\max\{t_\mathsf{ver}, t_\mathsf{sig}\}+1}$
and
$t_\mathsf{gen}$, $t_\mathsf{sig}$, $t_\mathsf{ver}$ are the
key generation time, signing time, and signature verification time, respectively.
\end{theorem}
\begin{proof}
Let $A$ be a $t$-time blocking adversary that succeeds with probability $\epsilon$. We construct a collision-finder $A_\mathsf{coll}$ and an
existential forger $A_\mathsf{ex}$ as follows:
\begin{itemize}
\item $A_\mathsf{coll}$ proceeds as follows:
 \begin{enumerate}
   \item $(\pubkey,\prikey)\gets \mathsf{G}$.
   \item Simulates $h_\mathsf{st}\gets A^{\mathsf{US},\mathsf{TS}(\prikey;\cdot)}(\pubkey)$ and records all the oracle queries.
   \item If $A^{\mathsf{US},\mathsf{TS}(\prikey;\cdot)}(\pubkey)$ was successful and either the case a) or b1) occurs, $A_\mathsf{coll}$ outputs the collision that is guaranteed in this case.
 \end{enumerate}
The oracles are simulated as follows:
 \begin{itemize}
   \item $\mathsf{US}$-queries: $A_\mathsf{coll}$ maintains its own version of $R$.
   \item $\mathsf{TS}(\prikey;\cdot)$-queries: $A_\mathsf{coll}$ uses the private key $\prikey$.
\end{itemize}
The computational time overhead function for the construction of $A_\mathsf{coll}$ is $\tau_\mathsf{coll}(t) = (\max\{t_\mathsf{ver}, t_\mathsf{sig}\}+1)\cdot t + t_\mathsf{gen}$, where
$t_\mathsf{ver}$ is the signature verification time (for $\unisrv$-queries),
$t_\mathsf{sig}$ is the signature creation time (for $\mathsf{TS}(\prikey;\cdot)$-queries), and $t_\mathsf{gen}$ is the key generation time.

\item $A_\mathsf{ex}^{\mathsf{S(\prikey;\cdot)}}(\pubkey)$ proceeds as follows:
 \begin{enumerate}
  \item Simulates $h_\mathsf{st}\gets A^{\mathsf{US},\mathsf{TS}(\prikey;\cdot)}(\pubkey)$ and records all the oracle queries.
  \item If $A^{\mathsf{US},\mathsf{TS}(\prikey;\cdot)}(\pubkey)$ was successful and b2) occurs and $Q=(\pubkey, h_\mathsf{st}, h_\mathsf{tx}, \sigma)$ was the request that enforces $R[H(\pubkey, h_\mathsf{st})]\neq\bot$ then:
  \begin{enumerate}
    \item[3.] $m\gets H(h_\mathsf{st}, h_\mathsf{tx})$.
    \item[4.] Output $(m,\sigma)$.
  \end{enumerate}
 \end{enumerate}
The oracles are simulated as follows:
 \begin{itemize}
  \item $\mathsf{US}$-queries are simulated so that $A_\mathsf{coll}$ maintains its own version of $R$.
  \item $\mathsf{TS}(\prikey;\cdot)$-queries are simulated by using calls to $\mathsf{S}(\prikey;\cdot)$.
 \end{itemize}
 As the request $Q$ was accepted by $\unisrv$, we have $\sigver(\pubkey,m,\sigma)=1$. Note that in the case b2) the request $\mathsf{S}(\prikey; m)$ was never made and hence, $A_\mathsf{ex}^{\mathsf{S(\prikey;\cdot)}}(\pubkey)$ is successful as an existential forger in the EF-CMA scenario. The computational time overhead function for the construction of $A_\mathsf{ex}$ is $\tau_\mathsf{ex}(t) = (t_\mathsf{ver}+1)\cdot t + t_\mathsf{hash}$, where
$t_\mathsf{ver}$ is the signature verification time (for $\unisrv$ queries) and
$t_\mathsf{hash}$ is the hash computation time (for output).
\end{itemize}
If $A$ succeeds, then either $A_\mathsf{coll}$ or $A_\mathsf{ex}$ succeeds and hence $\epsilon \le \epsilon_\mathsf{coll} + \epsilon_\mathsf{ex}$.
Assuming that $t_\mathsf{hash}\le t_\mathsf{gen}$, the inequality
$\tau_\mathsf{ex}(t)\le \tau_\mathsf{coll}(t)$ holds and hence by equation (\ref{eq:losebound})
\[
S_\mathsf{block}(\epsilon) = \tau^{-1}_\mathsf{coll}(S(\epsilon/2)) =
\frac{S(\epsilon/2)}{\max\{t_\mathsf{ver}, t_\mathsf{sig}\}+1} - \frac{t_\mathsf{gen}}{\max\{t_\mathsf{ver}, t_\mathsf{sig}\}+1}
\]
is a security profile of the Unicity Service against blocking.
\end{proof}

\subsection{Security against Double-Spending}\label{sec:double-spending}

A double-spending adversary uses $\unisrv$ as an oracle. \medskip

\noindent\textbf{Double-spending scenario} involves the following steps:
\begin{enumerate}
\item $(T,\sigma,\txhash,d,\pinc),(T',\sigma',\txhash',d',\pinc'),(\pubkey,h)\gets A^\unisrv$.
\item The attack is successful iff $T\neq T'$ and
\begin{equation}\label{eq:dscond}
\certver(T,\sigma,\txhash,d,\pinc;\pubkey,h)=\certver(T',\sigma',\txhash',d',\pinc';\pubkey,h)=1 \enspace.
\end{equation}
\end{enumerate}

\begin{definition}[Double-spending security]
The Unicity Service is said to be $S$-secure against double-spending if it has $S$ as a security profile in the double-spending scenario.
\end{definition}

\noindent\textbf{Analysis}: If the adversary is successful, then
from (\ref{eq:dscond}) and the definition of $\certver$ it follows that
$T.\sthash = T'.\sthash=h$ and:
\begin{eqnarray*}
\univer(H(\pubkey, h), \txhash; \pinc) = \univer(H(\pubkey, h), \txhash'; \pinc') = 1\enspace,
\end{eqnarray*}
which implies $\txhash=\txhash'$ by equation~(\ref{eq:eqtx}).
From Def.~\ref{de:certstate} it also follows that
$\open(\txhash,d)=H(T.D)$ and $\open(\txhash,d')=\open(\txhash',d')=H(T'.D)$.
From $(h,T.D)=(T.\sthash,T.D)=T\neq T'=(T'.\sthash,T'.D)=(h,T'.D)$ it follows that $T.D\neq T'.D$. Hence, we have two cases:
\begin{itemize}
\item[a)] $H(T.D)=H(T'.D)$, which
means that a collision has been found for $H$.
\item[b)] $H(T.D)\neq H(T'.D)$, which implies $\open(\txhash,d)=H(T.D)\neq H(T'.D)=\open(\txhash,d')$ and hence, the commitment $\txhash$ has been opened in two different ways.
\end{itemize}

\begin{theorem}
If $H$ is $S$-secure collision-resistant and the commitment scheme is $S$-secure computationally binding, then the Unicity service is $S_\mathsf{double}$-secure against double-spending, where
$S_\mathsf{double}(\epsilon) = \frac{S(\epsilon/2)}{t_\mathsf{ver}+1}$
and $t_\mathsf{ver}$ is the signature verification time.
\end{theorem}
\begin{proof}
Let $A$ be a $t$-time double-spending adversary that succeeds with probability $\epsilon$. We construct a collision-finder $A_\mathsf{coll}$ for the hash function and a double-opening adversary $A_\mathsf{com}$ for the commitment scheme as follows:
\begin{itemize}
\item $A_\mathsf{coll}$ proceeds as follows:
\begin{enumerate}
\item Simulate $(T,\sigma,\txhash,d,\pinc),(T',\sigma',\txhash',d',\pinc'),(\pubkey,h)\gets A^\unisrv$ by maintaining its own version of $\unisrv$.
\item Output the pair $(T.D,T'.D)$.
\end{enumerate}
The computational overhead function of $A_\mathsf{coll}$ is $\tau_\mathsf{coll}(t) = (t_\mathsf{ver}+1) \cdot t$, because simulating a $\unisrv$ query requires one signature verification and the number of calls is limited by the running time $t$ of $A$.
\item $A_\mathsf{com}$ proceeds as follows:
\begin{enumerate}
\item Simulate $(T,\sigma,\txhash,d,\pinc),(T',\sigma',\txhash',d',\pinc'),(\pubkey,h)\gets A^\unisrv$ by maintaining its own version of $\unisrv$.
\item Output the triple $(\txhash,d,d')$.
\end{enumerate}
The computational overhead function of $A_\mathsf{com}$ is the same as
that of $A_\mathsf{coll}$, i.e. $\tau_\mathsf{com}(t)= \tau_\mathsf{coll}(t)=(t_\mathsf{ver}+1) \cdot t$.
\end{itemize}
 If $A$ succeeds, then in case a) the collision finder $A_\mathsf{coll}$ succeeds, and in case b) the double-opener $A_\mathsf{com}$ succeeds. Hence, $\epsilon \le \epsilon_\mathsf{coll} + \epsilon_\mathsf{com}$, where $\epsilon_\mathsf{coll}$ is the success probability of $A_\mathsf{coll}$ and
$\epsilon_\mathsf{com}$ is the success probability of $A_\mathsf{com}$.
Therefore, by equation (\ref{eq:losebound}), the function $S_\mathsf{double}$ defined by
\[
S_\mathsf{double}(\epsilon) = \tau^{-1}_\mathsf{coll}(S(\epsilon/2)) =
\frac{S(\epsilon/2)}{t_\mathsf{ver}+1}
\]
is a security profile of the Unicity service against double-spending.
\end{proof}

\subsection{Insecure Modifications of the Unicity Service}

\subsubsection{State Hash not Signed}

Consider the following modification of $\mathsf{US}$ that, given a request $Q=(\pubkey, h_\mathsf{st}, h_T, \sigma)$ proceeds as follows:
\begin{enumerate}
\item If $R[H(\pubkey, h_\mathsf{st})]=\bot$ and $\sigver(\pubkey, h_T, \sigma)=1$ then $R[H(\pubkey, h_\mathsf{st})]\gets h_T$.
\item Return a proof $\pi$ of the statement $R[H(\pubkey, h_\mathsf{st})]=h_T$.
\end{enumerate}

\noindent Assume that a user $A$ owns a token in state $(\pubkey,\sthash)$.
A malicious user that knows any pair  $(h,\sigma)$ such that $\sigver(\pubkey,h,\sigma)=1$
can now lock $A$-s token by sending malicious request $Q=(\pubkey,\sthash, h, \sigma)$ to $\mathsf{US}$.
Other users may indeed know such pairs if they have received tokens from $A$ (and hence, having seen transactions $(T,\sigma)$ signed by $A$).
Hence, such a $\unisrv$ is insecure against blocking.

\subsubsection{$R$ as a Set}\label{rasaset}

Consider the following modification of $\mathsf{US}$ where $R$ is just a set and a request $Q=(\pubkey, \sthash, \txhash, \sigma)$ is processed as follows:
\begin{enumerate}
\item If $H(\pubkey, \sthash)\not\in R$ and $\mathsf{V}(\pubkey, H(\sthash, \txhash), \sigma)=1$ then $R\gets R \cup \{H(\pubkey, \sthash)\}$.
\item Return a proof $\pinc$ of the statement $H(\pubkey, \sthash)\in R$.
\end{enumerate}

\noindent The verification function $\univer$ just ignores the second argument, i.e.  $\univer(k,v,\pi)=1$ holds iff $k\in R$ in the current state of $\unisrv$. \medskip

\noindent A user $A$ who owns a token in state $(\pubkey,\sthash)$ (such that $H(\pubkey,\sthash)\not\in R$) can now proceed as follows:
\begin{enumerate}
\item $A$ creates two signed transactions $(T_1,\sigma_1)$ and $(T_2,\sigma_2)$ with
$T_1=(\sthash,D_1)$ and $T_2=(\sthash,D_2)$ with $T_1.D.\pubkey'\neq T_2.D.\pubkey'$, i.e. $T_1$ and $T_2$ transfer the same token to two different public keys.
\item Let $(\txhash^1,d_1)\gets \commit(H(D_1))$ and $(\txhash^2,d_2)\gets \commit(H(D_2))$.
\item $A$ calls $\pinc\gets\unisrv(Q)$, where $Q=(\pubkey,\sthash,\txhash^1,\sigma_1)$ and $\pinc$ is a proof of the statement $H(\pubkey,\sthash)\in R$. After that, $H(\pubkey,\sthash)\in R$ holds.
\item Also, both $(T_1,\sigma_1,\txhash^1,d_1, \pinc)$ and $(T_2,\sigma_2,\txhash^2,d_2, \pinc)$ are certified transactions in $(\pubkey,\sthash)$ because, as $H(\pubkey,\sthash)\in R$:
\[
\univer(H(\pubkey,\sthash),\txhash^1,\pinc)=
\univer(H(\pubkey,\sthash),\txhash^2,\pinc)=1\enspace.
\]
\end{enumerate}

\noindent Therefore, such a $\unisrv$ is insecure against double-spending.

\section{Service Side Privacy}\label{sec:privacy}

Unicity service $\unisrv$ obtains information about the transactions $T=(\sthash,D)$ with tokens via the queries $Q=(\pubkey, \sthash, \txhash,\sigma)$. We want to ensure that $\unisrv$ does not learn too much about the contents and context of transactions, for example, which transaction belongs to which token.

Assume that a token is currently in the state $(\pubkey,\sthash)$ and the next transaction with the token is $T=(\sthash, D)$, where $D=(\pubkey',x,\auxd')$. To certify $T$, the query
$Q=(\pubkey,\sthash,\txhash,\sigma)$ is sent to $\unisrv$ where $\txhash=\commitc(H(D))$ and $\sigma = \sig(\prikey; H(\sthash,\txhash))$.
Assume that $\unisrv$ stores the query $Q$.
In the future, the next transaction $T'=(\sthash',D')$ will be executed with the same token and the query
$Q'=(\pubkey',\sthash',\txhash',\sigma')$ with $\sthash'=H(\sthash,x)$ will be received by $\unisrv$.

We do not want $\unisrv$ to be able to associate $Q$ and $Q'$ as two consecutive transactions with the same token. Such association is possible if $\unisrv$ somehow obtains the random $x$ included in the transaction $T$, because $\unisrv$ can then check that $H(\sthash,x)=\sthash'=f_{\sthash}(x)$. There are several ways how to find $x$:
\begin{itemize}
\item Invert the function $f_{\sthash}(\cdot)=H(\sthash,\cdot)$, i.e. find $x'$ such that $f_{\sthash}(x')=\sthash'$ and hope that $x'=x$. To prevent that, we may assume that the hash function $H$ is $(k,\ell)$-one-way.
\item Find $x$ based on the commitment $\txhash=\commitc(D)$. To prevent that, we may assume that the commitment scheme is computationally hiding.
\item Combine both techniques, i.e. invert $f_{\sthash}(x)$ with additional information about $x$ obtained from $\txhash$. To prevent that, we assume that the commitment scheme in use is \emph{perfectly hiding}. We will give a proof later in this section under some reasonable assumptions.
\end{itemize}

\noindent Note that if $\ell$ is large, then $H(\sthash,x)$ with $x\gets \{0,1\}^\ell$ may give very little information about the previous state hash $\sthash$.
For example, an extreme case is that if $\ell=k$ and the function $H(h,\cdot)\colon \{0,1\}^k\rightarrow \{0,1\}^k$ happens to be one-to-one for every $h$ (which most likely never happens for practical hash functions), then in fact $H(\sthash,x)$ gives no information on $\sthash$  because the equation $H(h,x)=\sthash'$ can be (uniquely) solved for every state hash $h$ that $\unisrv$ has stored or memorized. In practice, there is no need to choose a very large $\ell$ as practical security is possible if $\ell$ is much smaller than $k$.

\subsection{Security against Association}\label{sec:association}

The Association adversary $A=(A_1,A_2)$ is two-stage. \medskip

\noindent\textbf{Association scenario} involves the following steps:
\begin{enumerate}
\item $(\sthash, \pubkey', \auxd', a)\gets A_1$.
\item $x\gets \{0,1\}^\ell$.
\item $\sthash'\gets H(\sthash,x)$.
\item $\txhash\gets\commitc(H(\pubkey',x,\auxd')))$.
\item $x'\gets A_2(a; \sthash',\txhash)$.
\item The attack is successful iff $\sthash\in\{0,1\}^k$, $x'\in\{0,1\}^\ell$, and $H(\sthash,x')=\sthash'$. The success $\epsilon$ of $A$ is the probability that the attack is successful.
\end{enumerate}

\begin{definition}[association security]
The Unicity Service is said to be $S$-secure against association if it has $S$ as a security profile in the association scenario.
\end{definition}

\begin{theorem}
If the hash function is $S$-secure $(k,\ell)$-one-way and the commitment scheme is perfectly hiding, then the Unicity Service is $S_\mathsf{assoc}$-secure against association, where $S_\mathsf{assoc}(\epsilon)= S(\epsilon) - t_\mathsf{sm} - t_\mathsf{hash} - t_\mathsf{com}$, where
$t_\mathsf{sm}$, $t_\mathsf{hash}$, $t_\mathsf{com}$ are the random sampling time, the hashing time, and the commitment computation time, respectively.
\end{theorem}
\begin{proof}
Let $A=(A_1,A_2)$ be a $t$-time adversary that succeeds in the association scenario with probability $\epsilon$. Consider the following modified attack scenario:
\begin{enumerate}
\item $(\sthash, \pubkey', \auxd', a)\gets A_1$.
\item $x\gets \{0,1\}^\ell$.
\item $\sthash'\gets H(\sthash,x)$.
\item $x''\gets \{0,1\}^\ell$.
\item $\txhash'\gets\commitc(H(\pubkey',x'',\auxd')))$.
\item $x'\gets A_2(a; \sthash',\txhash')$.
\item The attack is successful iff $\sthash\in\{0,1\}^k$, $x'\in\{0,1\}^\ell$, and $H(\sthash,x')=\sthash'$.
\end{enumerate}

\noindent For any fixed value of $L=(\sthash, \pubkey', \auxd', a)$, due to perfect hiding, commitments $\txhash=\commitc(H(\pubkey',x,\auxd'))$ and $\txhash'=\commitc(H(\pubkey',x'',\auxd'))$ have equal probability distributions. Moreover, by Lemma~\ref{le:outputindependence} (with $g(x)=H(\pubkey',x,\auxd')$), the random variables $x$ and $\txhash=\commitc(H(\pubkey',x,\auxd'))$ are independent.

Since $x''$ and $x$ are independent, the commitment $\txhash'=\commitc(H(\pubkey',x'',\auxd'))$ is independent of both $x$ and $\sthash'=H(\sthash,x)$. Therefore, the joint distributions of $(\sthash',\txhash)$ and $(\sthash',\txhash')$ are equal, and hence $A$ succeeds in the modified scenario with probability $\epsilon$. We construct an adversary $A'=(A'_1,A'_2)$ as follows:
\begin{itemize}
\item $A'_1$ proceeds as follows:
\begin{enumerate}
\item $(\sthash, \pubkey', \auxd', a)\gets A_1$;
\item return $(\sthash, a')$, where $a'=(\pubkey', \auxd', a)$.
\end{enumerate}
\item $A'_2(a';y)$ with $a'=(\pubkey', \auxd', a)$ proceeds as follows:
\begin{enumerate}
\item $x''\gets \{0,1\}^\ell$;
\item $\txhash'\gets\commitc(H(\pubkey',x'',\auxd')))$;
\item $x'\gets A_2(a; y,\txhash')$;
\item return $x'$.
\end{enumerate}
\end{itemize}

\noindent The computational time overhead function of $A'$ is $\tau(t) = t + t_\mathsf{sm} + t_\mathsf{hash} + t_\mathsf{com}$ and hence, by equation (\ref{eq:losebound}), the function $S_\mathsf{assoc}$ defined by
\[
S_\mathsf{assoc}(\epsilon) = \tau^{-1}(S(\epsilon)) = S(\epsilon) - t_\mathsf{sm} - t_\mathsf{hash} - t_\mathsf{com}
\]
is a security profile of the Unicity service against association.
\end{proof}

\subsection{Discussion}

The security proofs against blocking and double spending are also valid in the case when the commitment scheme is omitted from the construction and the transaction data hash is computed by $\txhash = H(\pubkey',x,\auxd')$.
At the same time, there seem to be no obvious ways of proving that such a simplified scheme is secure against association.

Still, it may be a plausible practical assumption that given $\sthash$ and the next state hash $\sthash'=H(\sthash,x)$ and $\txhash = H(\pubkey',x,\auxd')$ it is infeasible to find $x'$ so that $H(\sthash,x')=\sthash'$. At least for example, if $H$ is \textsf{SHA-256}, there seem to be no obvious ways how to do it if $x$ has sufficiently many random bits.

\subsection{Next Transaction Timing Privacy}

While the Unicity infrastructure prevents the Unicity Service from linking consecutive transactions with the same token, a different privacy issue emerges in direct peer-to-peer transfers.

\paragraph{The Privacy Leak}

Consider a transaction where party $A$ (the sender) transfers a token to party $B$ (the recipient). During this transfer, $A$ possesses the following information:
\begin{itemize}
\item The recipient's public key $\pubkey'$
\item The random nonce $x \gets \{0,1\}^\ell$ used in the transaction data
\item The current state hash $\sthash$ of the token
\end{itemize}

\noindent From this information, $A$ can compute the next state hash:
\[
\sthash' = H(\sthash, x) \enspace.
\]

\noindent When $B$ subsequently spends the token by executing a transaction, the Unicity Service records the mapping $R[H(\pubkey', \sthash')] \gets \txhash'$ for some transaction hash $\txhash'$. Since the registry $R$ is publicly accessible (anyone can query whether $R[k] = \bot$ for any key $k$), party $A$ can fetch the key $H(\pubkey', \sthash')$ to detect when $B$ spends the token.

This represents a timing privacy leak: although $A$ cannot determine the recipient of $B$'s subsequent transaction nor the transaction details, $A$ learns the \emph{moment} when $B$ chooses to spend the token. In contexts where transaction timing correlates with sensitive information (e.g., spending patterns revealing business activities), this leakage may be undesirable.

\paragraph{Impact} The practical significance of the leak is limited, as:
\begin{itemize}
\item The leak reveals only timing information, not transaction content or recipient identity
\item It requires the sender to actively and continuously monitor the registry
\item It applies only to the next direct transfer
\end{itemize}

\noindent Nevertheless, we present two mitigation strategies for contexts where even timing privacy is required.

\subsection{Solution 1: Timing Obfuscation}

The simplest mitigation strategy is for party $B$ to introduce uncertainty about the actual spending moment through a dummy transaction, e.g. to himself. Specifically:

\begin{itemize}
\item Upon receiving a token, $B$ may probabilistically execute a transaction to themselves
\item This creates a registry entry $R[H(\pubkey', \sthash')] \gets \txhash'$ at a time chosen by $B$
\item The token's state then transitions to $(\pubkey'', \sthash'')$ where $\pubkey'' = \pubkey'$ and $\sthash'' = H(\sthash', x')$ for a fresh random $x' \gets \{0,1\}^\ell$
\item Party $A$ may observe the registry update but cannot distinguish whether $B$ has genuinely spent the token or executed a dummy self-transfer
\end{itemize}

\noindent By randomizing the decision to make such self-transfers, $B$ can effectively hide the timing of the actual spending transaction. The trade-off is the cost of additional transactions to the Unicity Service and the increased transaction history length.

\subsection{Solution 2: Masked Address Protocol}

An interactive protocol that prevents $A$ from computing the observable registry key prevents such leak. The recipient $B$ computes a masking key $y$ which blinds the recipient's identity and provides the randomness needed for the next state hash. There is no need for $x$ from the base protocol.

\paragraph{Protocol}

\begin{enumerate}
\item \textbf{Recipient generates masked address}: Party $B$ samples a random mask $y \gets \{0,1\}^\ell$ and computes:
\[
\mpubkey' \gets H(\pubkey', y) \enspace.
\]

\item \textbf{Address exchange}: $B$ sends $\mpubkey'$ to $A$.

\item \textbf{Transaction creation}: $A$ creates transaction data $D = (\mpubkey', \auxd')$ (note: no $x$ field), computes the commitment $(\txhash, d) \gets \commit(H(D))$, creates signature $\sigma \gets \sig(\prikey, H(\sthash, \txhash))$, and obtains inclusion proof $\pinc \gets \unisrv(\pubkey, \sthash, \txhash, \sigma)$.

\item \textbf{Transaction transmission}: $A$ sends $(T, \sigma, \txhash, d, \pinc)$ to $B$. Note that $y$ need not be transmitted since $B$ already possesses it.

\item \textbf{State hash computation}: Upon receiving the certified transaction, $B$ computes the next state hash as:
\[
\sthash' \gets H(\sthash, y) \enspace.
\]

\item \textbf{Token ledger verification}: When $B$ presents the token to a third party or spends it, the ledger entry $(T, \sigma, \txhash, d, \pinc; \pubkey, \sthash)$ must be verified along with the tuple $(\pubkey', y)$ where:
\begin{itemize}
\item $H(\pubkey', y) = T.D.\mpubkey'$ (mask consistency)
\item $\sthash' = H(\sthash, y)$ (state hash derivation)
\end{itemize}

\item \textbf{Spending the token}: $B$ spends from state $(\pubkey', \sthash')$ by creating a new certification request $Q' = (\pubkey', \sthash', \txhash'', \sigma')$ to the Unicity Service.
\end{enumerate}

\paragraph{Security Against Double-Spending}

The simplified protocol maintains the fundamental double-spending prevention property. We must verify that $B$ cannot create multiple valid next states from a single certified transaction.

\begin{theorem}[No Parallel Registry Entries]
If $H$ is $S$-secure collision-resistant, then recipient $B$ cannot create two distinct valid next states from a single certified transaction.
\end{theorem}

\begin{proof}
Assume $B$ receives a certified transaction $(T, \sigma, \txhash, d, \pinc)$ in state $(\pubkey, \sthash)$, where $T.D.\mpubkey' = H(\pubkey', y)$ for $B$'s public key $\pubkey'$ and mask $y$.

Suppose $B$ attempts to create two distinct next states $(\pubkey', \sthash'_1)$ and $(\pubkey', \sthash'_2)$ with $\sthash'_1 \neq \sthash'_2$. For both states to be valid continuations from the certified transaction, there must exist masks $y_1, y_2$ such that:
\begin{align*}
H(\pubkey', y_1) &= T.D.\mpubkey' = H(\pubkey', y_2) \\
\sthash'_1 &= H(\sthash, y_1) \\
\sthash'_2 &= H(\sthash, y_2)
\end{align*}

From the first equation, if $H$ is collision-resistant, then with overwhelming probability $y_1 = y_2$, which implies $\sthash'_1 = H(\sthash, y_1) = H(\sthash, y_2) = \sthash'_2$, contradicting the assumption that $\sthash'_1 \neq \sthash'_2$.

More formally, if $B$ could create such distinct states, then $B$ could be used to construct a collision-finding adversary $A_{\mathsf{coll}}$ that outputs $((\pubkey', y_1), (\pubkey', y_2))$ with $y_1 \neq y_2$ but $H(\pubkey', y_1) = H(\pubkey', y_2)$, contradicting the collision-resistance of $H$.

Therefore, the certified transaction from $A$ uniquely determines $B$'s next state as $(\pubkey', H(\sthash, y))$ where $y$ is the unique preimage (up to collisions) of $T.D.\mpubkey'$ under the function $H(\pubkey', \cdot)$.
\end{proof}

\paragraph{Security Against Blocking}

The masked protocol inherits the blocking resistance of the standard protocol:

\begin{theorem}[No Blocking in Masked Address Protocol]
The masked protocol is $S_{\mathsf{block}}$-secure against blocking with the same security profile as the standard protocol.
\end{theorem}

\begin{proof}
The blocking security analysis from Section~\ref{sec:blocking} applies directly. The only modification is that transaction data now has the form $D = (\mpubkey', \auxd')$ instead of $D = (\pubkey', x, \auxd')$.

An adversary attempting to block state $(\pubkey', \sthash')$ must create a valid certification request $Q = (\pubkey', \sthash', \txhash, \sigma)$ where $\sigver(\pubkey', H(\sthash', \txhash), \sigma) = 1$. This requires either:
\begin{itemize}
\item Breaking the signature scheme (existential forgery), or
\item Finding a collision in $H$ to reuse a previous signature
\end{itemize}

Neither attack is facilitated by the removal of $x$ from the transaction structure. The security reduction proceeds identically to Theorem~1 in Section~\ref{sec:blocking}.
\end{proof}

\paragraph{Timing Privacy}

The masked protocol provides next transaction timing privacy:

\begin{theorem}[Timing Privacy in Masked Address Protocol]
If $H$ is $(k, \ell)$-one-way, party $A$ cannot feasibly determine when party $B$ spends the token.
\end{theorem}

\begin{proof}
For $A$ to monitor when $B$ spends the token, $A$ must compute the registry key $k' = H(\pubkey', \sthash')$ that will be set upon $B$'s certification request. This requires computing $\sthash' = H(\sthash, y)$.

Party $A$ knows $\sthash$ and $\mpubkey' = H(\pubkey', y)$ but not $y$. To compute $\sthash'$, $A$ must:
\begin{enumerate}
\item Extract $y$ from $\mpubkey'$ and $\pubkey'$ (which may not be known to $A$), or
\item Guess $y$ and verify the guess by checking if $H(\sthash, y)$ matches some observed registry entry
\end{enumerate}

Option (1) requires inverting $H(\pubkey', \cdot)$, contradicting the one-wayness of $H$. Option (2) fails because $A$ cannot determine $\pubkey'$ from $\mpubkey'$ (again by one-wayness) and thus cannot construct the registry key $H(\pubkey', H(\sthash, y))$ to verify any guess.
\end{proof}

\paragraph{Implementation Considerations}

The masked protocol requires $B$ to reliably store the mask $y$ along with the private key $\prikey'$ and state hash $\sthash'$ until making the next transaction. Loss of $y$ renders the token unspendable, as $B$ cannot prove that $\sthash' = H(\sthash, y)$ is the legitimate next state derived from the previous owner's certified transaction.

\subsection{Solution 3: Non-Interactive Masked Protocol}\label{sec:non-interactive-masked-protocol}

The interactive masked protocol requires the recipient to generate and communicate $\mpubkey' = H(\pubkey', y)$ before the sender can construct the transaction. We now present a non-interactive variant that eliminates this communication round while preserving all security properties.

\subsubsection{Mask Derivation via Verifiable Random Function}

We observe that a Verifiable Random Function (VRF) allows the holder of a private key to produce a unique, verifiable pseudorandom value derived from an input. This property enables verifiable, deterministic mask derivation without revealing the private key.

\paragraph{VRF Construction}

The recipient $B$ computes a \emph{VRF proof} and \emph{output} on the incoming state hash:
\[
(y, \pi) \gets \mathsf{VRF.Prove}(\prikey', \sthash)
\]

The value $y$ serves as the mask. As the VRF has the uniqueness property, $y$ is uniquely determined by $(\prikey', \sthash)$.

\paragraph{Complete Protocol}

\medskip
\noindent\textbf{Transaction $A \to B$}:

\begin{enumerate}
\item Transaction data: $D = (\pubkey', \auxd')$
\item Commitment: $(\txhash, d) \gets \commit(H(D))$
\item Signature: $\sigma \gets \sig(\prikey, H(\sthash, \txhash))$
\item Request to $\unisrv$: $Q = (\pubkey, \sthash, \txhash, \sigma)$
\item Obtain inclusion proof $\pinc$
\item Send to $B$: $(T, \sigma, \txhash, d, \pinc)$ where $T = (\sthash, D)$
\end{enumerate}

\medskip
\noindent\textbf{$B$ computes next state}:

\begin{enumerate}
\item VRF computation: $(y, \pi) \gets \mathsf{VRF.Prove}(\prikey', \sthash)$
\item Next state hash: $\sthash' \gets H(\sthash, y)$
\end{enumerate}

\medskip
\noindent\textbf{Transaction $B \to C$}:

\begin{enumerate}
\item Transaction data: $D' = (\pubkey'', \auxd'')$
\item Commitment: $(\txhash', d') \gets \commit(H(D'))$
\item Signature: $\sigma' \gets \sig(\prikey', H(\sthash', \txhash'))$
\item Request to $\unisrv$: $Q' = (\pubkey', \sthash', \txhash', \sigma')$
\item Obtain inclusion proof $\pinc'$
\item Send to $C$: $(T', \sigma', \txhash', d', \pinc'; \pi)$ along with ledger $\mathcal{L}$
\end{enumerate}

\paragraph{Recipient $C$'s Verification}

Given the certified transaction, binding signature $\sigma_{\mathsf{bind}}$, and ledger $\mathcal{L}$:

\begin{enumerate}
\item \textbf{Verify VRF}:
\[
\mathsf{VRF.Verify}(\pubkey', \sthash, y, \pi) \stackrel{?}{=} 1
\]
where $y$ is the claimed mask provided (or derived from $\pi$).

\item \textbf{Verify state hash derivation}:
\[
\sthash' \stackrel{?}{=} H(\sthash, y)
\]

\item \textbf{Verify certified transaction}:
\[
\certver(T', \sigma', \txhash', d', \pinc'; \pubkey', \sthash') \stackrel{?}{=} 1
\]

\item \textbf{Verify ledger}: Recursively verify all previous transactions in $\mathcal{L}$
\end{enumerate}

\paragraph{Security Analysis}

\begin{theorem}[No Double-Spending]
If the signature scheme is deterministic and unforgeable, recipient $B$ cannot create two distinct valid next states from a single certified transaction.
\end{theorem}

\begin{proof}
Given a certified transaction in state $(\pubkey, \sthash)$ transferring to $\pubkey'$, the VRF output $y$ is uniquely determined by $\sthash$ and the private key $\prikey'$ corresponding to $\pubkey'$.

For a secure VRF, there exists exactly one valid output $y$ for input $\sthash$ and key $\prikey'$. Therefore $\sthash' = H(\sthash, y)$ is unique.

To produce a different valid state $\sthash'' \neq \sthash'$, party $B$ would need to produce a valid proof $\pi'$ for a different output $y'$, which contradicts the uniqueness property of the VRF.
\end{proof}

\begin{theorem}[Timing Privacy]
If the VRF is pseudorandom, party $A$ cannot determine when $B$ spends the token.
\end{theorem}

\begin{proof}
To compute the registry key $H(\pubkey', \sthash')$ that $\unisrv$ will record when $B$ spends, party $A$ must compute $\sthash' = H(\sthash, y)$ where $y$ is the VRF output.

Party $A$ knows $\pubkey'$ and $\sthash$. However, computing $y$ requires $\prikey'$.

Since $A$ does not possess $\prikey'$, predicting $y$ contradicts the pseudorandomness property of the VRF.
\end{proof}

\begin{theorem}[Service-Side Privacy]
The Unicity Service cannot link consecutive token states.
\end{theorem}

\begin{proof}
The service $\unisrv$ observes requests $Q = (\pubkey, \sthash, \txhash, \sigma)$ and $Q' = (\pubkey', \sthash', \txhash', \sigma')$. To link these as consecutive states of the same token, $\unisrv$ must verify that $\sthash' = H(\sthash, y)$ for some valid $y$.

Computing the VRF output $y$ requires $\prikey'$. The service knows $\sthash$ but not $\prikey'$, and therefore cannot compute $y$.

Without $y$, the service cannot verify any relationship between $\sthash$ and $\sthash'$ due to the preimage resistance of $H$ and pseudorandomness of $y$.
\end{proof}

\paragraph{Implementation Considerations}

This protocol requires a VRF. A practical instantiation of a VRF is a \emph{Unique Signature Scheme}, where for every valid key pair $(pk, sk)$ and message $m$, there exists exactly one signature $\sigma$ such that $\sigver(pk, m, \sigma)=1$.

In this case:
\begin{enumerate}
\item $\mathsf{VRF.Prove}(sk, m)$ outputs $\pi = \sigma = \mathsf{Sign}(sk, m)$ and $y = H(\sigma)$.
\item $\mathsf{VRF.Verify}(pk, m, y, \pi)$ checks $\sigver(pk, m, \pi) \stackrel{?}{=} 1$ and $y \stackrel{?}{=} H(\pi)$.
\end{enumerate}

\noindent Recommended instantiations include:
\begin{itemize}
\item \textbf{BLS Signatures} (Boneh-Lynn-Shacham): BLS signatures are inherently unique.
\item \textbf{RSA-FDH} (Full Domain Hash): RSA signatures with a full-domain hash are unique.
\item Specially designed VRF like EC-VRF (RFC 9381).
\end{itemize}

\noindent Note that deterministic signature shemes like EdDSA and Schnorr are \textbf{not} unique signature schemes as-is and thus are unsafe for this protocol.

\subsection{Combined Protocol with Unlinkable Keys}

We now present the complete protocol combining \emph{unlinkable transaction keys} and \emph{persistent recipient identities} (Section~\ref{sec:non-interactive-protocol}) with the binding signature mechanism (Section~\ref{sec:non-interactive-masked-protocol}) for verifiable mask derivation.

\subsubsection{Key Derivation}

For a recipient with persistent keypair $(d, P = [d]G)$ and sender-chosen ephemeral value $r \in \mathbb{Z}_q$:

\begin{enumerate}
\item Sender computes ephemeral point: $R \gets [r]G$
\item Sender computes shared secret: $S \gets [r]P$
\item Blinding factor: $s \gets H_{\mathsf{key}}(S, R, \sthash)$
\item Transaction public key: $P_{\mathsf{tx}} \gets P + [s]G$
\item Recipient computes same shared secret: $S = [d]R$
\item Recipient derives private key: $d_{\mathsf{tx}} \gets d + s \mod q$
\end{enumerate}

Note that $[d_{\mathsf{tx}}]G = [d + s]G = [d]G + [s]G = P + [s]G = P_{\mathsf{tx}}$, so the derived keys are consistent.

\subsubsection{Transaction Protocol}

\paragraph{Transaction $A \to B$}

Party $A$ holds a token in state $(P_A, \sthash)$  (note the ECDSA specific notation, instead of $\pubkey_A$ we write $P_A$) and transfers to party $B$ with persistent public key $P_B$.

\begin{enumerate}
\item \textbf{Key derivation}:
\begin{itemize}
\item $A$ samples $r \gets \mathbb{Z}_q$ and computes $R \gets [r]G$
\item $A$ computes shared secret $S \gets [r]P_B$
\item $A$ computes blinding factor $s \gets H_{\mathsf{key}}(S, R, \sthash)$
\item $A$ computes transaction key $P_{\mathsf{tx}} \gets P_B + [s]G$
\end{itemize}

\item \textbf{Transaction data}: $D = (P_{\mathsf{tx}}, R, \auxd')$

\item \textbf{Commitment}: $(\txhash, d) \gets \commit(H(D))$

\item \textbf{Signature}: $\sigma \gets \sig(d_A, H(\sthash, \txhash))$

\item \textbf{Certification}: $Q = (P_A, \sthash, \txhash, \sigma)$, obtain $\pinc \gets \unisrv(Q)$

\item \textbf{Send to $B$}: $(T, \sigma, \txhash, d, \pinc)$ where $T = (\sthash, D)$
\end{enumerate}

\paragraph{$B$ Processes Incoming Transaction}

\begin{enumerate}
\item \textbf{Validate ephemeral point}: Verify $R \neq \mathcal{O}$ (point at infinity)

\item \textbf{Derive keys}:
\begin{itemize}
\item Compute shared secret $S \gets [d_B]R$
\item Compute blinding factor $s \gets H_{\mathsf{key}}(S, R, \sthash)$
\item Verify $P_B + [s]G = P_{\mathsf{tx}}$ (from transaction data)
\item Compute transaction private key $d_{\mathsf{tx}} \gets d_B + s \mod q$
\end{itemize}

\item \textbf{VRF computation}: $(y, \pi) \gets \mathsf{VRF.Prove}(d_{\mathsf{tx}}, \sthash)$

\item \textbf{Next state}: $\sthash' \gets H(\sthash, y)$

\item \textbf{Store}: $(d_{\mathsf{tx}}, \sthash', \pi)$ for spending
\end{enumerate}

\paragraph{Transaction $B \to C$}

Party $B$ spends the token to party $C$ with persistent public key $P_C$.

\begin{enumerate}
\item \textbf{Key derivation for $C$}:
\begin{itemize}
\item $B$ samples $r' \gets \mathbb{Z}_q$ and computes $R' \gets [r']G$
\item $B$ computes $S' \gets [r']P_C$
\item $B$ computes $s' \gets H_{\mathsf{key}}(S', R', \sthash')$
\item $B$ computes $P'_{\mathsf{tx}} \gets P_C + [s']G$
\end{itemize}

\item \textbf{Transaction data}: $D' = (P'_{\mathsf{tx}}, R', \auxd'')$

\item \textbf{Commitment}: $(\txhash', d') \gets \commit(H(D'))$

\item \textbf{Signature}: $\sigma' \gets \sig(d_{\mathsf{tx}}, H(\sthash', \txhash'))$

\item \textbf{Certification}: $Q' = (P_{\mathsf{tx}}, \sthash', \txhash', \sigma')$, obtain $\pinc' \gets \unisrv(Q')$

\item \textbf{Send to $C$}: $(T', \sigma', \txhash', d', \pinc'; \pi)$ along with ledger $\mathcal{L}$
\end{enumerate}

\subsubsection{Transaction Verification}

Given the certified transaction, binding signature $\sigma_{\mathsf{bind}}$, and ledger $\mathcal{L}$:

\begin{enumerate}
\item \textbf{Verify VRF}:
\[
\mathsf{VRF.Verify}(P_{\mathsf{tx}}, \sthash_{\mathsf{prev}}, y, \pi) \stackrel{?}{=} 1
\]
where $P_{\mathsf{tx}}$ is the transaction public key from the previous transaction's data.

\item \textbf{Verify state hash derivation}:
\[
\sthash' \stackrel{?}{=} H(\sthash_{\mathsf{prev}}, y)
\]

\item \textbf{Verify certified transaction}:
\[
\certver(T', \sigma', \txhash', d', \pinc'; P_{\mathsf{tx}}, \sthash') \stackrel{?}{=} 1
\]

\item \textbf{Verify ledger}: Recursively verify all previous transactions in $\mathcal{L}$
\end{enumerate}

\section{User Side Privacy}\label{sec:wallet-privacy}

In the previous sections, there is the assumption that the recipient generates a fresh keypair for every transaction. This may be impractical in some applications, especially where the recipient's secure storage is limited. We present two solutions that allow for the generation of unlinkable public keys while maintaining only a single persistent private key.

As a requirement, we avoid the persistent state on the client side that must be retained between transactions.

\subsection{A General Solution using Pseudo-Random Functions}

In this subsection, we define \emph{multi public key (MPK) signature schemes} where a user has a single private key but many public keys, each capable of verifying signatures created with the same private key. This is achieved by separating key generation into two phases:
\begin{enumerate}
\item Private key generation.
\item Public key generation that depends on a user-chosen index $\iota$.
\end{enumerate}
MPK signature schemes have to satisfy two security properties:
\begin{enumerate}
\item \emph{Existential unforgeability}, which is similar to the standard EF-CMA condition for the ordinary signature schemes.
\item \emph{Unlinkability} of keys, i.e. the one-time public keys and the signatures verifiable with one-time keys contain no useful information about whether they correspond to the same private key or not.
\end{enumerate}
We will show that a wide class of generic signature schemes $\mathsf{S}=(G,S,V)$ can
be transformed to MPK signature schemes $\mathsf{S}^F$ by using a pseudo random function family $F$. This wide class contains, for example all Schnorr type signature schemes (like EdDSA) and DSA type signature schemes like ECDSA (used in Bitcoin). In principle, such a transformation also applies to RSA but is much less efficient.
We will prove that whenever a generic scheme is EF-CMA secure and $F$ is a secure pseudo-random function family, then the MPK signature scheme $\mathsf{S}^F$ is both existentially unforgeable and unlinkable.

\subsubsection{Definition and Motivation}

\begin{definition}[MPK signature scheme]
A multi public key (MPK) signature scheme is a tuple $(G_s,G_p,S,V)$ such that:
\begin{itemize}
\item $k\gets G_s$ generates the private key
\item $X_\iota\gets G_p(k,\iota)$ generates a (one-time) public key for an index $\iota$
\item $\sigma\gets S(k;\iota,m)$ creates a signature for a message $m$
\item $V(X,m,\sigma)$ verifies a signature $\sigma$ on a message $m$ with public key $X$
\end{itemize}
so that the following verification identity holds for every key $k$, for every message $m$, and for every index $\iota$:
\[
V(G_p(k,\iota),m,S(k;\iota,m))=1\enspace.
\]
\end{definition}
MPK signature schemes may be used in payment scenarios, where a payee has the private key $k$ of a MPK signature scheme
$(G_s,G_p,S,V)$:
\begin{enumerate}
\item The payee first sends an invoice to the payer, where the invoice contains the index $\iota$ as a \emph{reference number} as well as the specific public key $X_\iota =G_p(k,\iota)$. The reference number may depend on the goods or services for which the invoice was sent.
\item The payer forms a transaction $T$ that among the other necessary information contains the reference number $\iota$, signs the transaction with a signature $\sigma$, and sends the signed transaction $(T,\sigma)$ to the payment system (e.g. Bitcoin blockchain).
\item The payee may then check that the signed transaction $(T,\sigma)$ was received by the payment system (for example, is in a sufficiently old Bitcoin block) and that
$T.X_\iota=G_p(k,T.\iota)$. By the explicit reference number $T.\iota$, the payee is also able to check for which goods or services the incoming payment was made.
\end{enumerate}
If the MPK signature scheme is unlinkable, it is not later possible to determine whether any two signed payment transactions $(T,\sigma), (T',\sigma')$
are associated with the same payee or the same payer.

\subsubsection{Security Definitions}

For the existential unforgeability, we use a formal security condition that is similar to the standard EF-CMA, but where the adversary can itself initiate public key generation and choose the corresponding indices $\iota$ arbitrarily by itself.
The attack scenario involves an oracle $\mathcal{O}^{\tilde{\mathsf{S}}}_k$ that
holds a secret key $k$, has a set $I$ (which is initially empty) as its state, and answers two types of queries:
\begin{itemize}
\item $\mathcal{O}^{\tilde{\mathsf{S}}}_k(\mathsf{gen};\iota)$ -- a key generation query that is answered as follows. If $\iota\in I$ the  oracle returns $\bot$ (the index was already used). Otherwise, the oracle sets $I\gets I\cup \{\iota\}$ (marks $\iota$ as "used") and returns $X_\iota = G_p(k,\iota)$.
\item $\mathcal{O}^{\tilde{\mathsf{S}}}_k(\mathsf{sig};\iota,m)$ -- a signing query that is answered as follows. If $\iota\not\in I$ the oracle returns $\bot$ (the key $X_\iota$ is not yet generated). Otherwise, the oracle returns $\sigma=S(k;\iota,m)$.
\end{itemize}

\begin{definition}[EF-CKMA]
A MPK signature scheme $\tilde{\mathsf{S}}=(G_s,G_p,S,V)$ is $s$-secure against existential forgeries under
chosen key-message attack (EF-CKMA) if it has $s$ as a security profile $s$ in the following attack scenario:
\begin{enumerate}
\item $k\gets G_s$
\item $(X,\iota,m,\sigma)\gets A^{\mathcal{O}^{\tilde{\mathsf{S}}}_k}$
\item The attack is successful iff:
\begin{enumerate}
\item $V(X;m,\sigma)=1$
\item The query $X\gets \mathcal{O}^{\tilde{\mathsf{S}}}_k(\mathsf{gen};\iota)$ was made by $A$
\item The query $\mathcal{O}^{\tilde{\mathsf{S}}}_k(\mathsf{sig};\iota,m)$ was never made by $A$
\end{enumerate}
The success $\epsilon$ of $A$ is the probability that $A$ is successful.
\end{enumerate}
\end{definition}
For defining the unlinkability, we introduce an additional oracle $\mathcal{O}^{\tilde{\mathsf{S}}}$ that is similar to the first oracle, but handles
queries as follows:
\begin{itemize}
\item $\mathcal{O}^{\tilde{\mathsf{S}}}(\mathsf{gen};\iota)$ -- a key generation query that is answered as follows. If $\iota\in I$, the oracle returns $\bot$ (the index was already used). Otherwise, the oracle sets $I\gets I\cup \{\iota\}$ (marks $\iota$ as "used"), generates $k_\iota\gets G_s$, saves $k_\iota$, and returns $X_\iota = G_p(k_\iota,\iota)$.
\item $\mathcal{O}^{\tilde{\mathsf{S}}}(\mathsf{sig};\iota,m)$ -- a signing query that is answered as follows. If $\iota\not\in I$ the oracle returns $\bot$ (the key $X_\iota$ is not yet generated). Otherwise, the oracle returns $\sigma=S(k_\iota;\iota,m)$.
\end{itemize}
\begin{definition}[Unlikability]\label{def:unlinkability}
A MPK signature scheme $\tilde{\mathsf{S}}=(G_s,G_p,S,V)$ is $s$-secure unlinkable if it has $s$ as a security profile in the following scenario:
\begin{enumerate}
\item $k\gets G_s$
\item $b_1\gets D^{\mathcal{O}^{\tilde{\mathsf{S}}}_k}$
\item $b_0\gets D^{\mathcal{O}^{\tilde{\mathsf{S}}}}$
\item The success of $D$ is
$\epsilon=|\mathsf{Pr}[b_1=1]-\mathsf{Pr}[b_0=1]|$.
\end{enumerate}
\end{definition}
The idea of the security condition is that a distinguisher $D$ cannot make difference whether the secret key is the same $k$ for all key-generation queries (the oracle $\mathcal{O}^{\tilde{\mathsf{S}}}_k$)
or is generated independently for every (adversary-chosen) index (the oracle $\mathcal{O}^{\tilde{\mathsf{S}}}$).

The oracles are defined so that the queries
$\mathcal{O}^{\tilde{\mathsf{S}}}_k(\mathsf{gen};\iota)$ cannot be made twice, i.e. $\bot$ is returned for a repeated query. This restriction is essential because otherwise $D$ can easily distinguish the oracles: $\mathcal{O}^{\tilde{\mathsf{S}}}_k(\mathsf{gen};\iota)$ will always be the same, but the public keys generated by calling $\mathcal{O}^{\tilde{\mathsf{S}}}(\mathsf{gen};\iota)$ twice can be different.

\subsubsection{EF-CKMA for Ordinary Signature Schemes}\label{sec:efckma-ordinary}
Before defining a general construction of MPK signature schemes from ordinary ones, we need to define EF-CKMA security for ordinary signature schemes. For that, we define an additional oracle $\mathcal{O}^{\mathsf{S}}$ that handles the
queries as follows:
\begin{itemize}
\item $\mathcal{O}^{\mathsf{S}}(\mathsf{gen};\iota)$ -- a key generation query that is answered as follows. If $\iota\in I$, the  oracle returns $\bot$ (the index was already used). Otherwise, the oracle sets $I\gets I\cup \{\iota\}$ (marks $\iota$ as "used"), generates $(x_\iota,X_\iota)\gets G$, saves $x_\iota$, and returns $X_\iota$.
\item $\mathcal{O}^{\mathsf{S}}(\mathsf{sig};\iota,m)$ -- a signing query that is answered as follows. If $\iota\not\in I$ the oracle returns $\bot$ (the key $X_\iota$ is not yet generated). Otherwise, the oracle returns $\sigma=S(x_\iota;m)$.
\end{itemize}

\begin{definition}[EF-CKMA for ordinary signatures]\label{def:efckma-ordinary}
A signature scheme $\mathsf{S}=(G,S,V)$ is $s$-secure EF-CKMA if it has $s$ as a security profile in the following attack scenario:
\begin{enumerate}
\item $(X,\iota,m,\sigma)\gets A^{\mathcal{O}^{\mathsf{S}}}$
\item The attack is successful iff:
\begin{enumerate}
\item $V(X;m,\sigma)=1$
\item The query $X\gets \mathcal{O}^{\mathsf{S}}(\mathsf{gen};\iota)$ was made by $A$.
\item The query $\mathcal{O}^{\mathsf{S}}(\mathsf{sig};\iota,m)$ was never made by $A$.
\end{enumerate}
\end{enumerate}
\end{definition}
EF-CKMA differs from traditional EF-CMA because the adversary is able to initiate the use of new keys and may create an existential forgery with any of the keys. The indices are used just for labeling the generated keys.

It turns out that EF-CKMA is not a conceptually stronger notion of security because any EF-CMA signature scheme is also EF-CKMA, though with some security loss.

\begin{theorem}[EF-CMA implies EF-CKMA]
Every signature scheme $\mathsf{S}=(G,S,V)$ that is $s$-secure EF-CMA is $s'$-secure EF-CKMA, where $s'$ is a function for which the implication $\tau(t)\ge s(\epsilon/t) \; \Rightarrow \; t\ge s'(\epsilon)$ holds for every $t$ and $\tau$ is defined by $\tau(t)=(1+\max\{t_\mathsf{gen},t_\mathsf{sig}\})\cdot t + t_\mathsf{sm}$.
\end{theorem}
\begin{proof}
Given a $t$-time adversary $A$ that has success $\epsilon$ in the EF-CKMA scenario, we construct an EF-CMA adversary $B$ that has success $\epsilon/t$ and running time $\tau(t)$, which means that $\tau(t)\ge s(\epsilon/t)$. The adversary $B^{S(x;\cdot)}(X)$ (where $(x,X)$ is a private-public key pair) proceeds as follows:
\begin{enumerate}
\item Generates a random number $t_0\gets \{1,\ldots,t\}$
\item Simulates $(\iota,X,m,\sigma)\gets A^{\mathcal{O}^\mathsf{S}}$ by answering the $\mathcal{O}^\mathsf{S}$-queries as follows:
\begin{itemize}
\item $\mathcal{O}^\mathsf{S}(\mathsf{gen};\iota)$ -- if this is the $t_0$-th query, set $\iota_0\gets \iota$ and return $X$. Otherwise, compute $(X_\iota,x_\iota)\gets G$ and return $X_\iota$.
\item $\mathcal{O}^\mathsf{S}(\mathsf{sig};\iota,m)$ -- If $\iota_0$ is defined and $\iota=\iota_0$, then query the oracle $\sigma\gets S(x;m)$ and return $\sigma$. Otherwise, if $x_\iota$ is not defined, return $\bot$; and if $x_\iota$ is defined, compute $\sigma \gets S(x_\iota;m)$ and return $\sigma$.
\end{itemize}
\item Returns $(m,\sigma)$.
\end{enumerate}
With probability $\epsilon/t$ the adversary $(\iota,X,m,\sigma)\gets A^{\mathcal{O}^\mathsf{S}}$ succeeds in the simulation so that $\iota=\iota_0$ which means that $B^{S(x;\cdot)}(X)$ succeeds with the same probability in the EF-CMA game. The computational overhead function can be chosen to be $\tau(t)=(1+\max\{t_\mathsf{gen},t_\mathsf{sig}\})\cdot t + t_\mathsf{sm}$.
\end{proof}

\medskip\medskip
\noindent To illustrate the security loss, assume that $s(\epsilon)=s_\mathsf{cma}\cdot \epsilon$ is a linear function. From $\tau(t)\ge s(\epsilon/t)$ it follows that:
\[
(t_\mathsf{sm}+\max\{t_\mathsf{gen},t_\mathsf{sig}\})\cdot t \ge
(1+\max\{t_\mathsf{gen},t_\mathsf{sig}\})\cdot t + t_\mathsf{sm} \ge s_\mathsf{cma}\cdot \epsilon/t
\]
Hence, $t\ge \sqrt{\frac{s_\mathsf{cma}\cdot \epsilon}{t_\mathsf{sm}+\max\{t_\mathsf{gen},t_\mathsf{sig}\}}}$ and hence it is sufficient to take \[s'(\epsilon)=\sqrt{\frac{s(\epsilon)}{t_\mathsf{sm}+\max\{t_\mathsf{gen},t_\mathsf{sig}\}}}\enspace.
\]

\subsubsection{Generic Construction and Security Proofs}

In this section, we define a wide class of \emph{MPK-friendly} signature schemes that can be transformed into MPK signature schemes.

\begin{definition}[MPK-friendly signature scheme]
A signature scheme is MPK-friendly if the key-generation function $G$ is in the form:
\begin{enumerate}
\item $x\gets \mathcal{X}$
\item $X\gets f(x)$
\item Output $(X,x)$
\end{enumerate}
where $f$ is a deterministic function.
\end{definition}
For example, ECDSA and EdDSA are MPK-friendly, because their keys are generated by the rule: $x\gets \mathbb{Z}_q$, $X=x\mathbf{g}$, where $\mathbf{g}$ is a generator of a cyclic group $\mathbb{G}$ with order $q$, i.e. in that case, $\mathcal{X}=\mathbb{Z}_q$ and $f\colon \mathbb{Z}_q\rightarrow \mathbb{G}$ is defined by  $f(x)=x\mathbf{g}$.

\begin{definition}[Transformed scheme]

Let $\mathsf{S}=(G,S,V)$ be a MPK-friendly signature scheme, and $F\colon K\times \mathbb{I} \rightarrow \mathcal{X}$ be a pseudo-random function family.
The transformed scheme $\mathsf{S}^F=(G_s,G_p,S',V)$ is defined as follows:
\begin{itemize}
\item $G_s$ outputs a uniformly random $k\gets K$
\item $G_p(k,\iota)=f(F_k(\iota))$ for every $k\in K$ and $\iota\in \mathbb{I}$
\item $S'(k;\iota,m)=S(F_k(\iota); m)$ for every $k\in K$, $\iota\in \mathbb{I}$, and for every message $m$
\end{itemize}
\end{definition}
It is easy to see that the transformed scheme is a MPK signature scheme. Note also that in the scenario $k \gets K, x_\iota \gets F_k(\iota), X_\iota\gets G_p(\iota, x_\iota)$
the distribution of $(X_\iota,x_\iota)$ is the same as in the scenario
$(X_\iota,x_\iota)\gets G$.

\begin{theorem}
If $\mathsf{S}=(G,S,V)$ is a MPK-friendly signature scheme and $F$ is an $s$-secure PRF, then the transformed scheme $\tilde{S}=\mathsf{S}^F=(G_s,G_p,S',V)$ is $s'$-secure unlinkable, where $s'(\epsilon) = \frac{s(\epsilon)}{1+\max\{t_\mathsf{gen},t_\mathsf{sig}\})}$.
\end{theorem}

\begin{proof}
Let $D$ be a $t$-time distinguisher that has success $\epsilon$ in the unlinkability scenario. We construct a distinguisher $\underline{D}$ for $F$ as follows. Given access to an oracle $\varphi(\cdot)$ the distinguisher $\underline{D}^{\varphi(\cdot)}$
proceeds as follows:
\begin{enumerate}
\item Simulates $b\gets D^{\mathcal{O}^{\varphi}}$, where the oracle calls are simulated as follows:
\begin{itemize}
\item $\mathcal{O}^{\varphi}(\mathsf{gen};\iota)$ -- If $\iota\in I$ then return $\bot$. Otherwise, set $I\gets I \cup \{\iota\}$ and return $f(\varphi(\iota))$
\item $\mathcal{O}^{\varphi}(\mathsf{sig};\iota,m)$ -- If $\iota\not\in I$ then return $\bot$. Otherwise, return $S(\varphi(\iota);m)$
\end{itemize}
\item Returns $b$
\end{enumerate}

If $\varphi(\cdot)=F_k(\cdot)$, then the oracle $\mathcal{O}^\varphi$ is equivalent to  $\mathcal{O}^{\tilde{\mathsf{S}}}_k$, because then:
\begin{itemize}
\item $\mathcal{O}^\varphi(\mathsf{gen};\iota) = \mathcal{O}^{\tilde{\mathsf{S}}}_k(\mathsf{gen};\iota)$ as $f(\varphi(\iota)) = f(F_k(\iota))=G_p(k,\iota)$, and
\item
$\mathcal{O}^\varphi(\mathsf{sig};\iota,m) = \mathcal{O}^{\tilde{\mathsf{S}}}_k(\mathsf{sig};\iota,m)$
due to $S(\varphi(\iota);m)=S(F_k(\iota);m)=S'(k;\iota,m)$
\end{itemize}
If $\varphi(\cdot)=\Phi(\cdot)$, then the oracle $\mathcal{O}^\varphi$ is equivalent to  $\mathcal{O}^{\tilde{\mathsf{S}}}$, because then:
\begin{itemize}
\item $\mathcal{O}^\varphi(\mathsf{gen};\iota) = \mathcal{O}^{\tilde{\mathsf{S}}}(\mathsf{gen};\iota)$
because
due to lazy sampling technique, computing $f(\varphi(\iota)) = f(\Phi(\iota))$
is equivalent to $k_\iota\gets K\equiv k_\iota\gets G_s$ and returning
$X_\iota = G_p(k_\iota,\iota)$
\item
$\mathcal{O}^\varphi(\mathsf{sig};\iota,m) = \mathcal{O}^{\tilde{\mathsf{S}}}_k(\mathsf{sig};\iota,m)$
due to $S(\varphi(\iota);m)=S(\Phi(\iota);m)=S'(k_\iota;\iota,m)$ in case $k_\iota$ was a saved value during lazy sampling and hence $\Phi(\iota)=k_\iota$.
\end{itemize}
Therefore, the unlinkability scenario with $D$ is equivalent to the PRF scenario with $\underline{D}$ and
hence, the success of $\underline{D}$ in the PRF-scenario is $\epsilon$. As the computational time overhead function is $\tau(t) = (1+\max\{t_\mathsf{gen},t_\mathsf{sig}\})\cdot t$, the statement of the theorem follows.
\end{proof}

\begin{theorem}
If $\mathsf{S}=(G,S,V)$ is a MPK-friendly signature scheme that is $s$-secure EF-CKMA and its transformed scheme $\tilde{S}=\mathsf{S}^F=(G_s,G_p,S',V)$ is $s$-secure unlinkable, then $\tilde{S}$ is $s'$-secure EF-CKMA, where $s'(\epsilon) = s(\epsilon/2)$.
\end{theorem}

\begin{proof}
Let $A^\mathcal{O}$ be a $t$-time oracle adversary that succeeds with probability $\epsilon$ in the EF-CKMA scenario. We construct a distinguisher $D^\mathcal{O}$ for unlinkability and an EF-CKMA adversary $A^\mathcal{O}_\mathsf{ckma}$ for $\mathsf{S}$ as follows:
\begin{itemize}
\item $D^\mathcal{O}$ simulates $(X,\iota,m,\sigma)\gets A^{\mathcal{O}}$ and returns 1 iff $A$ is successful.
\item $A_\mathsf{ckma}^\mathcal{O}$ also simulates $(X,\iota,m,\sigma)\gets A^{\mathcal{O}}$ and returns $(X,\iota,m,\sigma)$.
\end{itemize}
In the EF-CKMA scenario, $k\gets K \equiv k\gets G_s$ is picked uniformly at random and $(X,\iota,m,\sigma)\gets A^{\mathcal{O}^{\tilde{\mathsf{S}}}_k}$ produces existential forgery with probability $\epsilon$. Hence, $\mathsf{Pr}[1\gets D^{\mathcal{O}^{\tilde{\mathsf{S}}}_k}]=\epsilon$. Let $\epsilon'$ be the probability that $(X,\iota,m,\sigma)\gets A^{\mathcal{O}^{\tilde{\mathcal{S}}}}$
produces existential forgery. Hence,
$\mathsf{Pr}[1\gets D^{\mathcal{O}^{\tilde{\mathsf{S}}}}]=\epsilon'$
and the success of $D$ in the unlinkability scenario is $\epsilon_\mathsf{unl}=|\epsilon-\epsilon'|$.
Note also that the oracles $\mathcal{O}^{\tilde{\mathsf{S}}}$ (from Def.~\ref{def:unlinkability}) and $\mathcal{O}^{\mathsf{S}}$ (from Def.~\ref{def:efckma-ordinary}) are equivalent because:
\begin{itemize}
\item $\mathcal{O}^{\tilde{\mathsf{S}}}(\mathsf{gen};\iota) = \mathcal{O}^{\mathsf{S}}(\mathsf{gen};\iota)$ as the scenario
[$k_\iota \gets G_s, X_\iota \gets G_p(k_\iota)$, Return $X_\iota$]
is equivalent to the scenario
[$(x_\iota, X_\iota)\gets G$, Return $X_\iota$].
\item $\mathcal{O}^{\tilde{\mathsf{S}}}(\mathsf{sig};\iota,m) = \mathcal{O}^{\mathsf{S}}(\mathsf{sig};\iota,m)$ as $S'(k_\iota;\iota,m)=S(F_{k_\iota}(\iota);m)=S(x_\iota;m)$.
\end{itemize}
Hence, $A_\mathsf{ckma}^{\mathcal{O}^\mathsf{S}}$
succeeds in the EF-CKMA scenario with probability $\epsilon_\mathsf{ckma}=\epsilon'$.
Therefore, $\epsilon \le \epsilon_\mathsf{unl}+ \epsilon_\mathsf{kcma}$. As the construction of $D$ and $A_\mathsf{ckma}$ involves no computational time overhead, we conclude that
$\tilde{\mathsf{S}}$ is $s'$-secure EF-CKMA with $s'(\epsilon) = s(\epsilon/2)$.
\end{proof}

\subsubsection{Remarks on Practical Implementation}

The main tool in the construction is the pseudo-random function family $F$ and hence the security of the scheme directly depends on the security of $F$ as a PRF.

For the ECDSA and EdDSA with groups the order $q$ of which is much less than $2^{512}$ we may apply the HMAC construction \cite{BeCK96} as follows:
\[
F_k(\iota) = \mathsf{Int}(\mathsf{HMAC}_{\mathsf{sha512}} (k;\iota)) \mod q
\]
where $\mathsf{Int}\colon \{0,1\}^{512}\rightarrow \mathbb{Z}_{2^{512}}$ converts bitstrings to non-negative integers, $k\gets K=\{0,1\}^{m}$, and $\iota\in\mathbb{I}=\{0,1\}^{64}$.
If
$q\ll 2^{512}$, then the output distribution
of $F_k(\cdot)$ is almost as indistinguishable
from the uniform distribution $U_q$ on $\mathbb{Z}_{q}$ as the output distribution of $\mathsf{HMAC}_{\mathsf{sha512}} (k;\cdot)$ from the uniform distribution $U_{2^{512}}$ on $\{0,1\}^{512}$
because the statistical distance between
$U_q$ and the distribution  $\mathsf{Int}(U_{2^{512}})\!\!\mod q$ has an upper bound $\frac{q}{2^{513}}$.

In \cite{Bell15}, Bellare proved that HMAC is a PRF if the compression function (of the hash function) is a PRF, which is a good practical security guarantee since no efficient attacks are known against the PRF-security of the compression function.

\subsection{A Solution for ECDSA}

Unlike the rest of the paper which is not signature scheme specific, we are assuming the ECDSA signature scheme\cite{ecdsa} in the following subsections, and use the ECDSA, DH specific notation.

\subsubsection{Interactive Protocol}

The recipient holds a persistent keypair $(d, P)$ where $P = d \cdot G$ with generator $G$ and order $n$. For each transaction, the recipient generates an ephemeral blinding factor $r \gets \mathbb{Z}_n$ and derives a transaction-specific public key $P' = (d + r) \cdot G$. The protocol is shown in Figure~\ref{fig:protocol1}.

\begin{figure}[!htb]
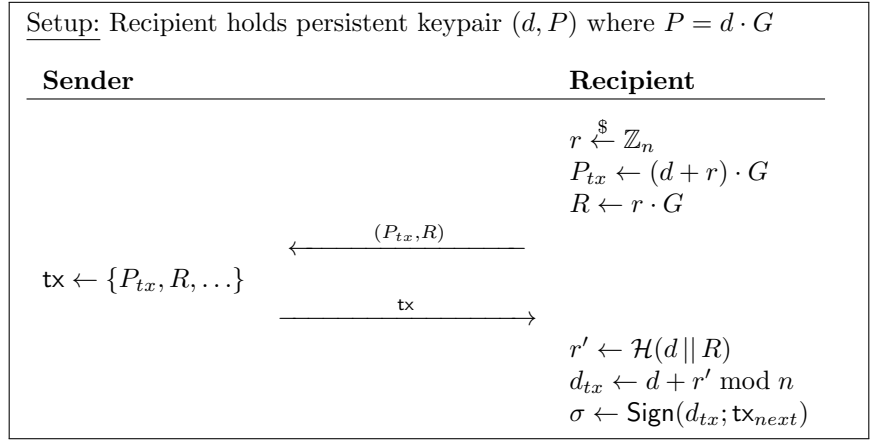

\centering
\fbox{
\begin{minipage}{0.9\textwidth}

\underline{Setup:} Recipient holds persistent keypair $(d, P)$ where $P = d \cdot G$

\vspace{3mm}
\begin{tabular}{lcl}
\textbf{Sender} & & \textbf{Recipient}\\
\hline
\\[-2mm]
  & &$r \stackrel{\$}{\gets} \mathbb{Z}_n$  \\
  & & $P_{tx} \gets (d + r) \cdot G$  \\
  & & $R \gets r \cdot G$  \\
  & $\xleftarrow{\hspace{1cm} (P_{tx}, R) \hspace{1cm}}$ & \\
$\mathsf{tx} \gets \{P_{tx}, R, \ldots\}$  & &  \\
  & $\xrightarrow{\hspace{1.5cm} \mathsf{tx} \hspace{1.5cm}}$ & \\
  & & $r' \gets \mathcal{H}(d \,||\, R)$ \\
  & & $d_{tx} \gets d + r' \bmod n$ \\
  & & $\sigma \gets \mathsf{Sign}(d_{tx}; \mathsf{tx}_{next})$ \\
\end{tabular}
\end{minipage}
}
\caption{Interactive protocol for generating unlinkable ECDSA public keys.}
\label{fig:protocol1}
\end{figure}

The recipient maintains only the persistent secret $d$ and derives ephemeral signing keys deterministically from data field $R$, included with the transaction. When spending the token, the recipient reconstructs the blinding factor $r' = H(d \,||\, R)$ and derives the private key $d' = d + r' \bmod n$ corresponding to $P'$.

\subsubsection{Non-Interactive Protocol}

For applications requiring persistent public keys as ``addresses'', or non-interactive operation, the parties can create the blinded public key using Diffie-Hellman key exchange. The recipient publishes a persistent public key $P = d \cdot G$. The protocol is shown in Figure~\ref{fig:protocol3}.

The key challenge in non-interactive protocols is protecting against malicious senders who might choose predictable ephemeral keys $r$ or leak them to compromise transaction unlinkability. The secure construction addresses this by binding the blinding factor to both the shared Diffie-Hellman secret and public transaction data.

\begin{figure}[ht]
\centering
\fbox{
\begin{minipage}{0.95\textwidth}

\underline{Setup:}
Recipient has persistent keypair $(d, P = d \cdot G)$ and publishes $P$

\vspace{3mm}
\begin{tabular}{lcl}
\textbf{Sender} & & \textbf{Recipient} \\
\hline
\\[-2mm]
$r \stackrel{\$}{\leftarrow} \mathbb{Z}_n^*$ & & \\
$R \leftarrow r \cdot G$ & & \\
$s \leftarrow \mathcal{H}_1(r \cdot P \,||\, R \,||\, \mathsf{tx}_{prev})$ & & \\
$P_{tx} \leftarrow P + s \cdot G$ & & \\
$\mathsf{tx} \leftarrow \{P_{tx}, R, \dots\}$ & & \\
& $\xrightarrow{\hspace{.6cm} \mathsf{tx} \hspace{.6cm}}$ & \\[3mm]
& & \textbf{if} $R = \mathcal{O}$ \textbf{then abort} \\
& & $s' \leftarrow \mathcal{H}_1(d \cdot R \,||\, R \,||\, \mathsf{tx}_{prev})$ \\
& & \textbf{if} $P + s' \cdot G \neq P_{tx}$ \textbf{then abort} \\
& & $d_{tx} \leftarrow d + s' \bmod n$ \\
& & \textit{Sign next transaction using} $d_{tx}$ \\
\end{tabular}
\end{minipage}
}
\caption{Non-interactive protocol secure against malicious sender.}
\label{fig:protocol3}
\end{figure}

The sender generates an ephemeral keypair $(r, R = r \cdot G)$ and computes a Diffie-Hellman shared secret $r \cdot P$. The blinding factor $s$ is derived by hashing the shared secret together with the ephemeral public key $R$ and the previous transaction identifier $\mathsf{tx}_{prev}$:
$$s = \mathcal{H}_1(r \cdot P \,||\, R \,||\, \mathsf{tx}_{prev})$$

The transaction-specific public key is computed as $P_{tx} = P + s \cdot G$, and both $P_{tx}$ and $R$ are included in the transaction. Upon receiving the transaction, the recipient:
\begin{enumerate}
\item Verifies that $R \neq \mathcal{O}$ (the point at infinity) to prevent trivial attacks
\item Computes the same shared secret $d \cdot R = r \cdot P$ using their persistent private key
\item Derives $s' = \mathcal{H}_1(d \cdot R \,||\, R \,||\, \mathsf{tx}_{prev})$ and verifies that $P + s' \cdot G = P_{tx}$
\item Computes the transaction-specific private key $d_{tx} = d + s' \bmod n$ for signing the next transaction
\end{enumerate}

By including both $R$ and $\mathsf{tx}_{prev}$ in the hash input, the protocol ensures that:
\begin{itemize}
\item Even if the sender chooses a predictable $r$, the blinding factor $s$ depends on the hash function output and remains unpredictable to external observers
\item The recipient can verify that the sender correctly computed $P_{tx}$ without learning $r$
\item Each transaction uses a unique blinding factor (assuming $\mathsf{tx}_{prev}$ is always unique), preventing linkability even if the sender reuses the same $r$ across different transactions
\end{itemize}

We assume that $\mathsf{tx}_{prev}$ is always unique, as a malicious sender reusing the exact $(r, \mathsf{tx}_{prev})$ pair could break unlinkability. In practice, $\mathsf{tx}_{prev}$ can be the hash of the previous transaction or a timestamp with sufficient granularity.

\subsubsection{Desired Privacy Properties:}

\begin{enumerate}
\item \textbf{Transaction Identity Unlinkability:}
For any two transactions $\mathsf{tx}_i, \mathsf{tx}_j$ with $i \neq j$:
$$\text{Dist}[(P_{tx,i}, R_i)] \approx_c \text{Dist}[(U_1, U_2)]$$
where $U_1, U_2 \stackrel{\$}{\leftarrow} \mathbb{G}$ are uniform random group elements.

\textit{Informally:} The ownership-related fields in each transaction appear as uniformly random data. The pair $(P_{tx}, R)$ in any transaction is computationally indistinguishable from two random group elements, making all transactions unlinkable to each other and to any recipient addresses.

\item \textbf{Forward Privacy:}
Given $d_{tx,i} = d_i + \mathcal{H}_1(d_i \cdot R_i \,||\, R_i \,||\, \mathsf{tx}_{i-1})$:
$$\Pr\left[\mathcal{A}(d_{tx,i}, \{\mathsf{tx}_j\}_{j=1}^n) \rightarrow d_i\right] \leq \text{negl}(\lambda)$$

\textit{Informally:} If a transaction-specific private key $d_{tx,i}$ is compromised (e.g., leaked or extracted from a device), the adversary cannot recover the recipient's persistent private key $d_i$. This ensures that past key exposures do not compromise the recipient's persistent identity or future transactions.
\end{enumerate}

\textbf{Remark:} Transaction Identity Unlinkability implies several important corollaries:
\begin{itemize}
\item \emph{Address Unlinkability:} An adversary observing a transaction cannot determine which of two candidate recipients it was sent to, beyond random guessing. This follows immediately since if $(P_{tx}, R)$ is pseudorandom, it reveals no information about the underlying persistent address $P$.

\item \emph{Recipient Unlinkability:} An adversary observing multiple transactions cannot determine which transactions belong to the same recipient. Even when a recipient reuses their persistent public key $P$ across multiple transactions, the resulting transaction-specific keys appear independent and uniformly random.
\end{itemize}

\subsubsection{Cryptographic Assumptions}

For the privacy proofs, we require the following standard cryptographic assumptions in the elliptic curve group $\mathbb{G}$ of prime order $n$ with generator $G$:

\begin{definition}[Decisional Diffie-Hellman (DDH)]
The DDH problem is $S$-hard in $\mathbb{G}$ if it has $S$ as a security profile in the following distinguishing scenario:
\begin{enumerate}
\item Sample $a, b, c \stackrel{\$}{\leftarrow} \mathbb{Z}_n^*$ and $d \stackrel{\$}{\leftarrow} \{0,1\}$
\item If $d = 0$: set $T \leftarrow (a \cdot G, b \cdot G, ab \cdot G)$
\item If $d = 1$: set $T \leftarrow (a \cdot G, b \cdot G, c \cdot G)$
\item $d' \gets \mathcal{A}(T)$
\item The attack is successful iff $d' = d$
\end{enumerate}
The advantage is $\text{Adv}^{\text{DDH}}_{\mathcal{A}} = \left|\Pr[d' = d] - \frac{1}{2}\right|$.
\end{definition}

\begin{definition}[Discrete Logarithm (DL)]
The DL problem is $S$-hard in $\mathbb{G}$ if it has $S$ as a security profile in the following attack scenario:
\begin{enumerate}
\item $x \stackrel{\$}{\leftarrow} \mathbb{Z}_n^*$, $X \leftarrow x \cdot G$
\item $x' \gets \mathcal{A}(X)$
\item The attack is successful iff $x' = x$
\end{enumerate}
\end{definition}

We model $\mathcal{H}_1 \colon \mathbb{G} \times \mathbb{G} \times \mathcal{M} \to \mathbb{Z}_n$ as a hash function with standard collision-resistance and additional properties suitable for key derivation.

\subsubsection{Privacy Proofs}

\begin{theorem}[Transaction Identity Unlinkability]
\label{thm:transaction-unlinkability}
If the DDH problem is $S$-hard in $\mathbb{G}$ and $\mathcal{H}_1$ is modeled as a random oracle, then the protocol in Fig.~\ref{fig:protocol3} satisfies transaction identity unlinkability. Specifically, for any two transactions $\mathsf{tx}_i, \mathsf{tx}_j$ with $i \neq j$:
$$\text{Dist}[(P_{tx,i}, R_i)] \approx_c \text{Dist}[(U_1, U_2)]$$
where $U_1, U_2 \stackrel{\$}{\leftarrow} \mathbb{G}$ are uniformly random, with distinguishing advantage at most $\text{Adv}^{\text{DDH}}_{\mathcal{A}} + \frac{q_H}{n}$ where $q_H$ is the number of hash queries.
\end{theorem}

\begin{proof}
For transaction $\mathsf{tx}_i$, the sender chooses $r_i \stackrel{\$}{\leftarrow} \mathbb{Z}_n^*$ and computes:
\begin{align*}
R_i &= r_i \cdot G \\
s_i &= \mathcal{H}_1(r_i \cdot P_i \,||\, R_i \,||\, \mathsf{tx}_{i-1}) \\
P_{tx,i} &= P_i + s_i \cdot G
\end{align*}

We use a game-based argument to show that $(P_{tx,i}, R_i)$ is indistinguishable from $(U_1, U_2)$ where $U_1, U_2 \stackrel{\$}{\leftarrow} \mathbb{G}$.

\textbf{Game 0:} The real transaction generation process as described above.

\textbf{Game 1:} Replace $s_i = \mathcal{H}_1(r_i \cdot P_i \,||\, R_i \,||\, \mathsf{tx}_{i-1})$ with a uniformly random $s_i \stackrel{\$}{\leftarrow} \mathbb{Z}_n^*$.

\noindent\textit{Claim:} $|\Pr[\text{Game 0}] - \Pr[\text{Game 1}]| \leq \text{Adv}^{\text{DDH}}_{\mathcal{A}} + \frac{q_H}{n}$.

Consider the tuple $(P_i, R_i, r_i \cdot P_i) = (d_i \cdot G, r_i \cdot G, r_i d_i \cdot G)$. Under the DDH assumption, this is computationally indistinguishable from $(d_i \cdot G, r_i \cdot G, c \cdot G)$ where $c \stackrel{\$}{\leftarrow} \mathbb{Z}_n^*$. Since $\mathcal{H}_1$ is modeled as a random oracle and $\mathsf{tx}_{i-1}$ is unique (by assumption), if the adversary does not query $\mathcal{H}_1(r_i \cdot P_i \,||\, R_i \,||\, \mathsf{tx}_{i-1})$, the value $s_i$ is uniformly random. The probability of guessing the correct input is at most $\frac{q_H}{n}$.

\textbf{Game 1 Analysis:} Since $r_i$ is chosen uniformly at random, $R_i = r_i \cdot G$ is uniformly distributed in $\mathbb{G}$. Given that $s_i \stackrel{\$}{\leftarrow} \mathbb{Z}_n^*$ is uniformly random and independent, $P_{tx,i} = P_i + s_i \cdot G$ is uniformly distributed in $\mathbb{G}$, since adding a uniformly random scalar times the generator to any group element yields a uniformly random element.

Therefore, in Game 1, $(P_{tx,i}, R_i)$ has the same distribution as two independent uniformly random group elements $(U_1, U_2)$.

Since this holds for each transaction independently (with fresh randomness $r_i$ for each transaction), the identity data in all transactions is unlinkable.
\end{proof}

\begin{theorem}[Forward Privacy]
\label{thm:forward-privacy}
If the discrete logarithm problem is $S$-hard in $\mathbb{G}$ and $\mathcal{H}_1$ is modeled as a random oracle, then for any PPT adversary $\mathcal{A}$:
$$\Pr\left[\mathcal{A}(d_{tx,i}, \{\mathsf{tx}_j\}_{j=1}^n) \rightarrow d_i\right] \leq \text{Adv}^{\text{DL}}_{\mathcal{A}} + \frac{q_H}{n}$$
where $d_{tx,i} = d_i + s_i \bmod n$ and $s_i = \mathcal{H}_1(d_i \cdot R_i \,||\, R_i \,||\, \mathsf{tx}_{i-1})$.
\end{theorem}

\begin{proof}
The recipient's transaction-specific private key is:
$$d_{tx,i} = d_i + s_i \bmod n$$
where $s_i = \mathcal{H}_1(d_i \cdot R_i \,||\, R_i \,||\, \mathsf{tx}_{i-1})$.

To recover $d_i$ from $d_{tx,i}$, the adversary must compute $s_i$, since:
$$d_i = d_{tx,i} - s_i \bmod n$$

Computing $s_i$ requires evaluating $\mathcal{H}_1(d_i \cdot R_i \,||\, R_i \,||\, \mathsf{tx}_{i-1})$, which in turn requires computing the shared secret $d_i \cdot R_i = d_i r_i \cdot G$.

The adversary observes:
\begin{itemize}
\item $R_i = r_i \cdot G$ (from the transaction)
\item $P_i = d_i \cdot G$ (possibly inferred or public, though unlinkable by Theorem~\ref{thm:transaction-unlinkability})
\item $d_{tx,i}$ (the leaked transaction key)
\end{itemize}

\textbf{Case 1:} The adversary does not query $\mathcal{H}_1(d_i \cdot R_i \,||\, R_i \,||\, \mathsf{tx}_{i-1})$.

Since $\mathcal{H}_1$ is a random oracle, without querying the correct input, $s_i$ appears uniformly random in $\mathbb{Z}_n$. Given $d_{tx,i} = d_i + s_i \bmod n$ where $s_i$ is unknown and uniform, the value $d_i$ is information-theoretically hidden. The adversary can succeed only by guessing among $q_H$ hash queries, with probability at most $\frac{q_H}{n}$.

\textbf{Case 2:} The adversary queries $\mathcal{H}_1(d_i \cdot R_i \,||\, R_i \,||\, \mathsf{tx}_{i-1})$.

To make this query, the adversary must compute $d_i \cdot R_i = d_i r_i \cdot G$ from $P_i = d_i \cdot G$ and $R_i = r_i \cdot G$. This is precisely the Computational Diffie-Hellman (CDH) problem, which is at least as hard as the discrete logarithm problem. If the adversary can compute $d_i \cdot R_i$, they can use the discrete logarithm oracle to extract $d_i$ from $P_i$ (or $r_i$ from $R_i$), succeeding with probability at most $\text{Adv}^{\text{DL}}_{\mathcal{A}}$.

Combining both cases, the adversary's advantage in recovering $d_i$ is bounded by $\text{Adv}^{\text{DL}}_{\mathcal{A}} + \frac{q_H}{n}$, which is negligible.
\end{proof}

\subsubsection{Remarks on Standard Model}

The proofs above use $\mathcal{H}_1$ as a random oracle. To achieve standard model security, we propose the following adjustment:

\textbf{Adjustment:} Replace the hash function $\mathcal{H}_1$ with a key derivation function (KDF) based on a pseudorandom function (PRF) family $\{f_k\}_{k \in \mathcal{K}}$ combined with randomness extraction. Specifically:
\begin{enumerate}
\item Use a randomness extractor $\text{Ext} \colon \mathbb{G} \to \{0,1\}^\lambda$ based on a strong randomness extractor (e.g., leftmost bits of the $x$-coordinate under suitable encoding).
\item Define $s \leftarrow \text{PRF}(\text{Ext}(r \cdot P), R \,||\, \mathsf{tx}_{prev})$ where PRF is a pseudorandom function.
\end{enumerate}

Under the DDH assumption, the value $r \cdot P$ is pseudorandom (given $r \cdot G$ and $P$), and thus $\text{Ext}(r \cdot P)$ serves as a pseudorandom key for the PRF. The PRF output is then pseudorandom, providing similar guarantees to the random oracle model.

With this adjustment, the proofs can be adapted to the standard model by replacing random oracle arguments with PRF indistinguishability under DDH, at the cost of slightly more complex security reductions and stronger assumptions on the randomness extractor properties.

\pagebreak
\appendix
\section{Detailed Transaction Sequence Diagram}

\begin{figure}[!h]
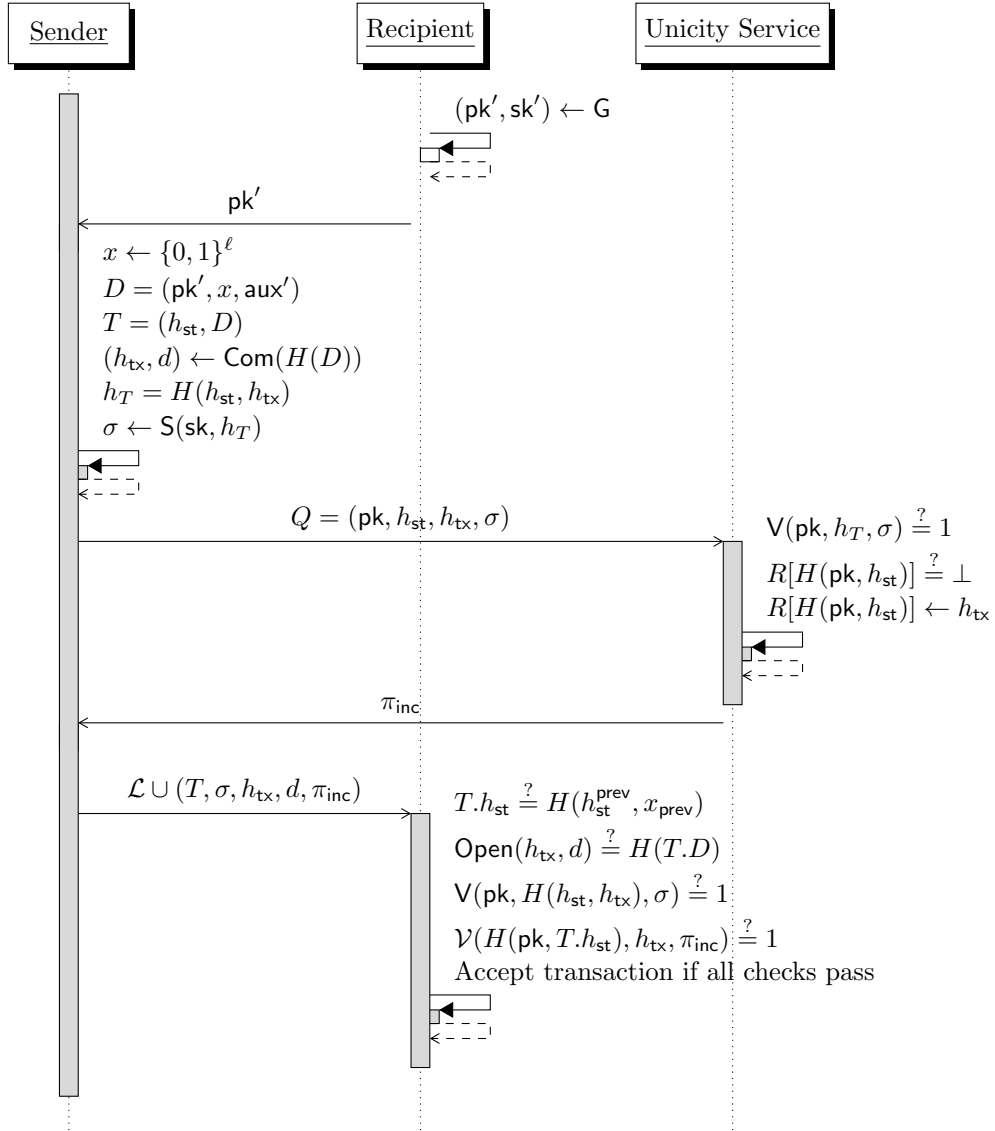

    \begin{center}
        \begin{sequencediagram}
            \newthread{S}{Sender}
            \newinst[3]{R}{Recipient}
            \newinst[2]{US}{Unicity Service}

            \begin{call}{R}{$(\pubkey', \prikey') \gets \keygen$}{R}{}
            \end{call}

            \begin{messcall}{R}{$\pubkey'$}{S}
            \end{messcall}

            \postlevel
            \postlevel
            \postlevel
            \begin{call}{S}{\shortstack[l]{
                    $x \gets \{0,1\}^\ell$     \\
                    $D = (\pubkey', x, \auxd')$ \\
                    $T = (\sthash, D)$         \\
                    $(\txhash, d) \gets \commit(H(D))$ \\
                    $h_T = H(\sthash, \txhash)$ \\
                    $\sigma \gets \sig(\prikey, h_T)$
                }
                }{S}{}
            \end{call}

            \begin{messcall}{S}{$Q = (\pubkey, \sthash, \txhash, \sigma)$}{US}
                \postlevel
                \begin{call}{US}{\shortstack[l]{
                    $\sigver(\pubkey, h_T, \sigma) \stackrel{?}{=} 1$ \\
                    $R[H(\pubkey, \sthash)] \stackrel{?}{=} \bot$ \\
                    $R[H(\pubkey, \sthash)] \gets \txhash$
                    }}{US}{}
                \end{call}
            \end{messcall}
            \prelevel
            \begin{messcall}{US}{$\pinc$}{S}
            \end{messcall}

            \begin{messcall}{S}{$\mathcal{L} \cup (T, \sigma, \txhash, d, \pinc)$}{R}
                \postlevel
                \postlevel
                \postlevel
                \begin{call}{R}{\shortstack[l]{
                    $T.\sthash \stackrel{?}{=} H(h^\textsf{prev}_\mathsf{st},x_\textsf{prev})$ \\
                    $\open(\txhash, d) \stackrel{?}{=} H(T.D)$ \\
                    $\sigver(\pubkey, H(\sthash, \txhash), \sigma) \stackrel{?}{=} 1$ \\
                    $\univer(H(\pubkey, T.\sthash), \txhash, \pinc) \stackrel{?}{=} 1$ \\
                    Accept transaction if all checks pass
                    }}{R}{}
                \end{call}
            \end{messcall}
        \end{sequencediagram}
        \caption{Detailed Unicity transaction flow.}\label{fi:unicity-transaction-full}
    \end{center}
\end{figure}

\begin{figure}[!h]
    \begin{center}
        \begin{sequencediagram}
            \newthread{S}{Sender $A$}
            \newinst[2.5]{R}{Recipient $B$}
            \newinst[2]{US}{Unicity Service}

            \postlevel
            \begin{call}{R}{\shortstack[l]{
                $y \gets \{0,1\}^\ell$     \\
                $\mpubkey' \gets H(\pubkey', y)$
            }}{R}{}
            \end{call}

            \begin{messcall}{R}{$\mpubkey'$}{S}
            \end{messcall}

            \postlevel
            \postlevel
            \begin{call}{S}{\shortstack[l]{
                    $D = (\mpubkey', \auxd')$ \\
                    $T = (\sthash, D)$ \\
                    $(\txhash, d) \gets \commit(H(D))$ \\
                    $\sigma \gets \sig(\prikey, H(\sthash, \txhash))$
                }}{S}{}
            \end{call}

            \begin{messcall}{S}{$Q = (\pubkey, \sthash, \txhash, \sigma)$}{US}
                \postlevel
                \begin{call}{US}{\shortstack[l]{
                    $\sigver(\pubkey, H(\sthash, \txhash), \sigma) \stackrel{?}{=} 1$ \\
                    $R[H(\pubkey, \sthash)] \stackrel{?}{=} \bot$ \\
                    $R[H(\pubkey, \sthash)] \gets \txhash$
                    }}{US}{}
                \end{call}
            \end{messcall}
            \prelevel
            \begin{messcall}{US}{$\pinc$}{S}
            \end{messcall}

            \begin{messcall}{S}{$(T, \sigma, \txhash, d, \pinc)$}{R}
                \postlevel
                \postlevel
                \begin{call}{R}{\shortstack[l]{
                    $\certver(T, \sigma, \txhash, d, \pinc; \pubkey, \sthash) \stackrel{?}{=} 1$ \\
                    $H(\pubkey', y) \stackrel{?}{=} T.D.\mpubkey'$ \\
                    $\sthash' \gets H(\sthash, y)$
                    }}{R}{}
                \end{call}
            \end{messcall}
        \end{sequencediagram}
        \caption{Masked Address Protocol (Interactive). Recipient generates masked address $\mpubkey' = H(\pubkey', y)$ to prevent sender from inferring next state.}\label{fi:masked-protocol}
    \end{center}
\end{figure}

\begin{figure}[!h]
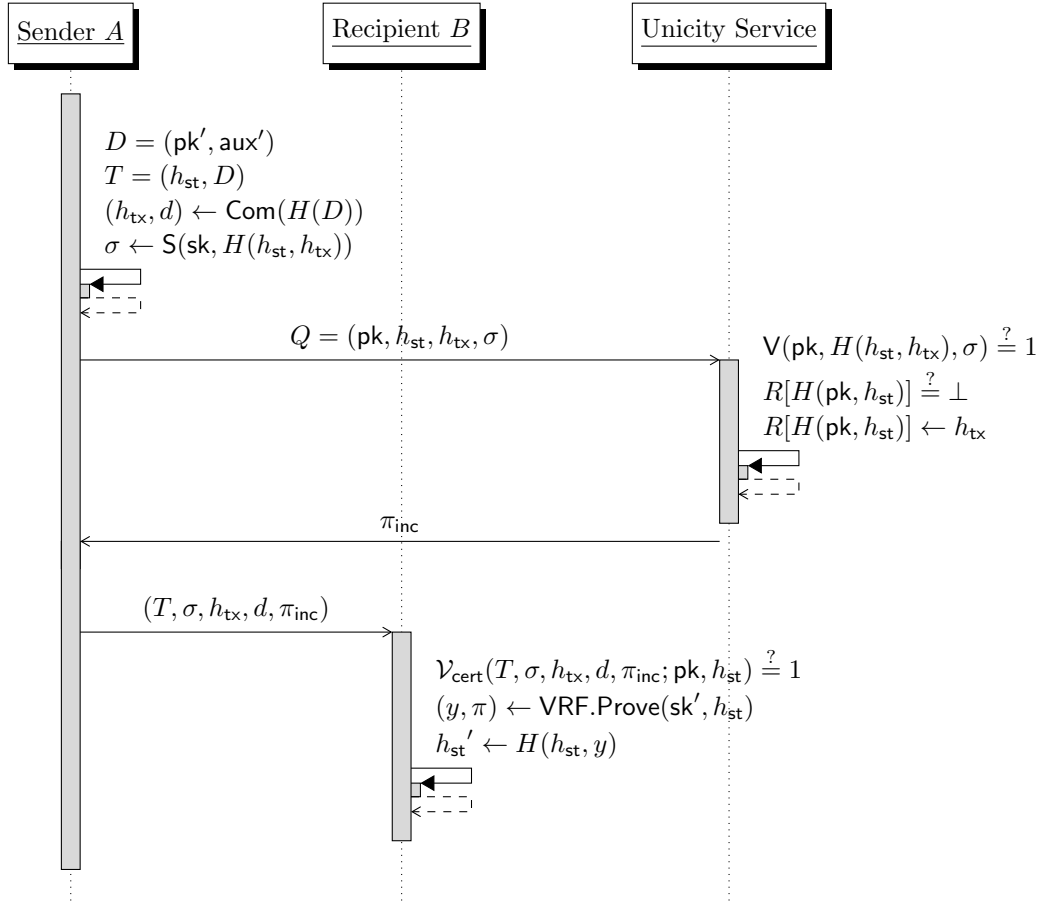

    \begin{center}
        \begin{sequencediagram}
            \newthread{S}{Sender $A$}
            \newinst[2.5]{R}{Recipient $B$}
            \newinst[2]{US}{Unicity Service}

            \postlevel
            \postlevel
            \postlevel
            \begin{call}{S}{\shortstack[l]{
                    $D = (\pubkey', \auxd')$ \\
                    $T = (\sthash, D)$ \\
                    $(\txhash, d) \gets \commit(H(D))$ \\
                    $\sigma \gets \sig(\prikey, H(\sthash, \txhash))$
                }}{S}{}
            \end{call}

            \begin{messcall}{S}{$Q = (\pubkey, \sthash, \txhash, \sigma)$}{US}
                \postlevel
                \begin{call}{US}{\shortstack[l]{
                    $\sigver(\pubkey, H(\sthash, \txhash), \sigma) \stackrel{?}{=} 1$ \\
                    $R[H(\pubkey, \sthash)] \stackrel{?}{=} \bot$ \\
                    $R[H(\pubkey, \sthash)] \gets \txhash$
                    }}{US}{}
                \end{call}
            \end{messcall}
            \prelevel
            \begin{messcall}{US}{$\pinc$}{S}
            \end{messcall}

            \begin{messcall}{S}{$(T, \sigma, \txhash, d, \pinc)$}{R}
                \postlevel
                \postlevel
                \begin{call}{R}{\shortstack[l]{
                    $\certver(T, \sigma, \txhash, d, \pinc; \pubkey, \sthash) \stackrel{?}{=} 1$ \\
                    $(y, \pi) \gets \mathsf{VRF.Prove}(\prikey', \sthash)$ \\
                    $\sthash' \gets H(\sthash, y)$
                    }}{R}{}
                \end{call}
            \end{messcall}
        \end{sequencediagram}
        \caption{Non-Interactive Masked Protocol. Recipient derives deterministic mask $y$ via VRF without prior interaction.}\label{fi:non-interactive-masked}
    \end{center}
\end{figure}

\begin{figure}[!h]
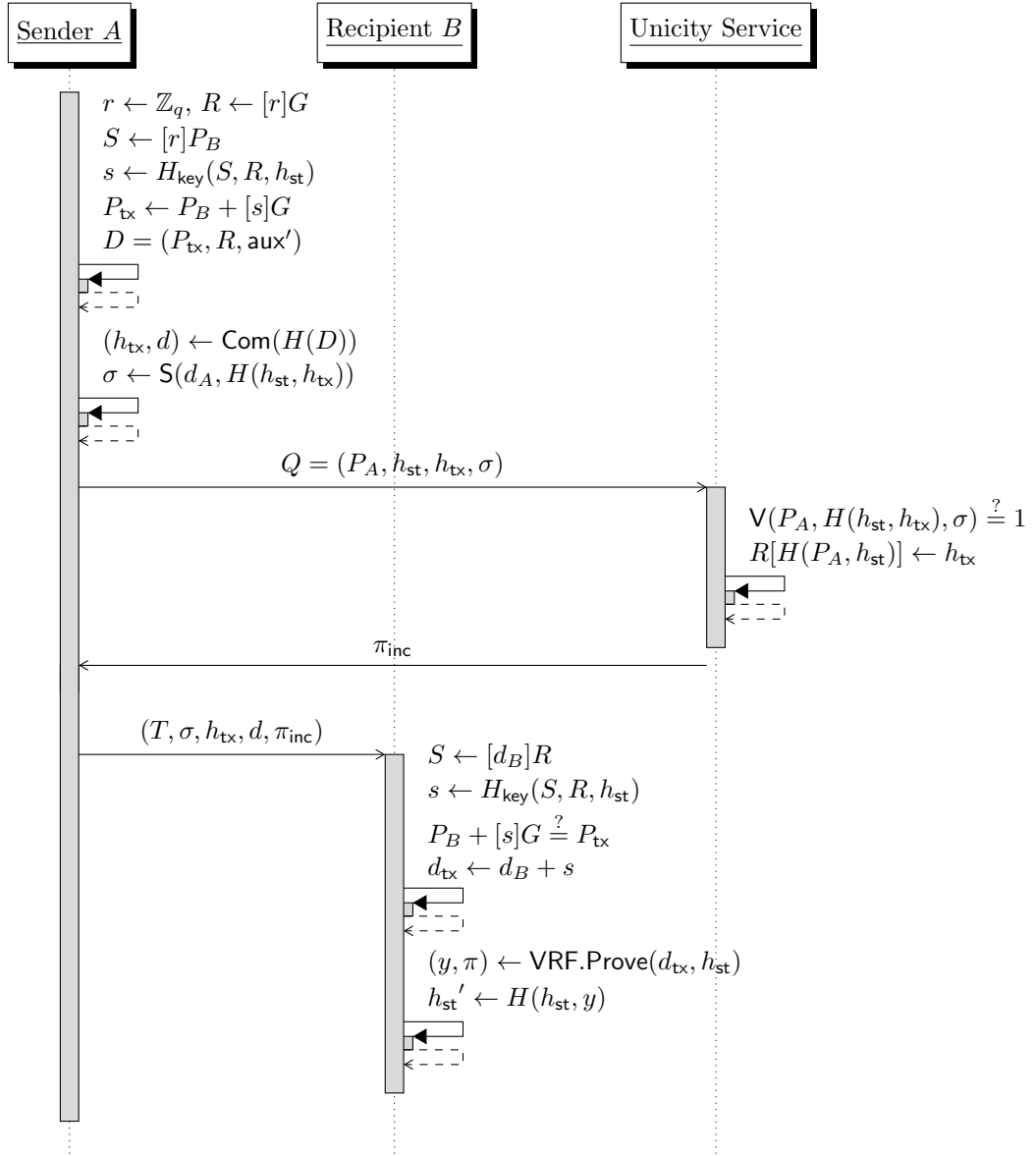

    \begin{center}
        \begin{sequencediagram}
            \newthread{S}{Sender $A$}
            \newinst[2.5]{R}{Recipient $B$}
            \newinst[2]{US}{Unicity Service}

            \postlevel
            \postlevel
            \postlevel
            \begin{call}{S}{\shortstack[l]{
                    $r \gets \mathbb{Z}_q$, $R \gets [r]G$ \\
                    $S \gets [r]P_B$ \\
                    $s \gets H_{\mathsf{key}}(S, R, \sthash)$ \\
                    $P_{\mathsf{tx}} \gets P_B + [s]G$ \\
                    $D = (P_{\mathsf{tx}}, R, \auxd')$
                }}{S}{}
            \end{call}

            \postlevel
            \begin{call}{S}{\shortstack[l]{
                    $(\txhash, d) \gets \commit(H(D))$ \\
                    $\sigma \gets \sig(d_A, H(\sthash, \txhash))$
                }}{S}{}
            \end{call}

            \begin{messcall}{S}{$Q = (P_A, \sthash, \txhash, \sigma)$}{US}
                \postlevel
                \begin{call}{US}{\shortstack[l]{
                    $\sigver(P_A, H(\sthash, \txhash), \sigma) \stackrel{?}{=} 1$ \\
                    $R[H(P_A, \sthash)] \gets \txhash$
                    }}{US}{}
                \end{call}
            \end{messcall}
            \prelevel
            \begin{messcall}{US}{$\pinc$}{S}
            \end{messcall}

            \begin{messcall}{S}{$(T, \sigma, \txhash, d, \pinc)$}{R}
                \postlevel
                \postlevel
                \begin{call}{R}{\shortstack[l]{
                    $S \gets [d_B]R$ \\
                    $s \gets H_{\mathsf{key}}(S, R, \sthash)$ \\
                    $P_B + [s]G \stackrel{?}{=} P_{\mathsf{tx}}$ \\
                    $d_{\mathsf{tx}} \gets d_B + s$
                    }}{R}{}
                \end{call}
                \postlevel
                \begin{call}{R}{\shortstack[l]{
                    $(y, \pi) \gets \mathsf{VRF.Prove}(d_{\mathsf{tx}}, \sthash)$ \\
                    $\sthash' \gets H(\sthash, y)$
                    }}{R}{}
                \end{call}
            \end{messcall}
        \end{sequencediagram}
        \caption{Combined Protocol with Unlinkable Keys. Recipient's permanent public key is $P_B$. Sender derives unlinkable transaction key $P_{\mathsf{tx}} = P_B + [s]G$ using ECDH. Recipient derives private key and mask via VRF.}\label{fi:combined-protocol}
    \end{center}
\end{figure}

\end{document}